%% file: SYCOS_TKDE2019.tex
\documentclass[10pt,journal,compsoc]{IEEEtran}

\ifCLASSOPTIONcompsoc
\usepackage[nocompress]{cite}
\else
\usepackage{cite}
\fi

\ifCLASSINFOpdf
\else
\fi
\usepackage{cite}
\usepackage[numbers,sort&compress]{natbib}
\usepackage{amsmath,amssymb,amsfonts}
\usepackage{textcomp}
\def\BibTeX{{\rm B\kern-.05em{\sc i\kern-.025em b}\kern-.08em
		T\kern-.1667em\lower.7ex\hbox{E}\kern-.125emX}}
\usepackage{tabularx}
\usepackage{graphicx}
\usepackage{algorithm}
\usepackage{algpseudocode}
\usepackage{subcaption}
\usepackage{threeparttable}
\usepackage{amsthm}
\usepackage{commath}
\usepackage{makecell}
\usepackage{array}
\usepackage{multirow}

\newcommand\setrow[1]{\gdef\rowmac{#1}#1\ignorespaces}
\newcommand\clearrow{\global\let\rowmac\relax}
\clearrow

\usepackage[dvipsnames]{xcolor}
\usepackage{pgfplots,pgfplotstable}
\usepackage{datatool}
\usepackage{readarray}
\usepackage{listofitems}
\readarraysepchar{,}

\usetikzlibrary{calc}

\definecolor{blizzardblue}{rgb}{0.78, 1, 1}
\definecolor{blue-violet}{rgb}{0.55, 0.49, 0.89}

\newtheorem{theorem}{Theorem}
\newtheorem{lem}{Lemma}

\newcommand{\Cross}{\mathbin{\tikz [x=1.8ex,y=1.8ex,line width=.2ex] \draw (0,0) -- (1,1) (0,1) -- (1,0);}}%

\pgfplotsset{
	compat=1.16,
	stylebar/.style={
		ybar,
		ylabel={ Runtimes},
		xlabel={Number of nodes},
		ymin=0,
		symbolic x coords={1,2,4,6,8,10,12,14,16,18,20},
		xtick=data,
		ylabel near ticks,
		nodes near coords,
		nodes near coords align={vertical},
		every node near coord/.append style={font=\footnotesize},
		bar width=3mm,
		legend style={
			at={(0.5,1.05)},    
			anchor=south,       
			legend columns=2,
		},
	},
	styleline/.style={
		xmin = 1,
		xmax = 20,
		axis y line=right,
		axis x line=none,
		symbolic x coords={1,2,4,6,8,10,12,14,16,18,20},
		xtick=data,
		enlarge x limits=0.1,
		ylabel = {Speed up},
		ylabel near ticks,
		ylabel style={rotate=-180},
		xmajorgrids,
		scaled y ticks = false,
		ymin=0, ymax=20,
		scatter/position=absolute,
		node near coords style={
			at={(\pgfkeysvalueof{/data point/x},0.3)},
			anchor=center,
			font=\footnotesize,
			/pgf/number format/.cd,
			fixed,
			precision=2,
			zerofill,
		},
	},
	customlegend/.style={
		ultra thick,
		lineplot
	},
	lineplot/.style={blue,mark=*,sharp plot,line legend}
}

\begin{document}
	%
	\title{A Unified Approach for Multi-Scale Synchronous Correlation Search in Big Time Series}
	\author{\IEEEauthorblockN{Nguyen Ho\IEEEauthorrefmark{1},
			Van Long Ho\IEEEauthorrefmark{1},
			Torben Bach Pedersen\IEEEauthorrefmark{1},
			Mai Vu\IEEEauthorrefmark{2},
			Christophe A.N. Biscio\IEEEauthorrefmark{3}} \\
		\IEEEauthorblockA{\IEEEauthorrefmark{1}Department of Computer Science, Aalborg University, Denmark}\\
		\IEEEauthorblockA{\IEEEauthorrefmark{2}Department of Electrical \& Computer Engineering, Tufts University, Medford, MA, USA} \\
		\IEEEauthorblockA{\IEEEauthorrefmark{3}Department of Mathematical Science, Aalborg University, Denmark}\\
		Email: ntth@cs.aau.dk, vlh@cs.aau.dk, tbp@cs.aau.dk, mai.vu@tufts.edu,  christophe@math.aau.dk}

	
	%

	\IEEEtitleabstractindextext{%
		\begin{abstract}
		The wide deployment of IoT sensors has enabled the collection of very big time series across different domains, from which advanced analytics can be performed to find unknown relationships, most importantly the {\em correlations} between them. However, current approaches for correlation search on time series are limited to only a single temporal scale and simple types of relations, and cannot handle noise effectively. This paper presents the {\em integrated SYnchronous COrrelation Search (iSYCOS)} framework to find multi-scale correlations in big time series. Specifically, iSYCOS integrates top-down and bottom-up approaches into a single auto-configured framework capable of efficiently extracting complex window-based correlations from big time series using {\em mutual information (MI)}. Moreover, iSYCOS includes a novel MI-based theory to identify noise in the data, and is used to perform pruning to improve iSYCOS performance. Besides, we design a {\em distributed version of iSYCOS} that can scale out in a Spark cluster to handle big time series. Our extensive experimental evaluation on synthetic and real-world datasets shows that iSYCOS can auto-configure on a given dataset to find {\em complex multi-scale correlations}. The pruning and optimizations can improve iSYCOS performance up to {\em an order of magnitude}, and the distributed iSYCOS can {\em scale out linearly} on a computing cluster. 
		\end{abstract}
		
		}
	
	\maketitle

	\input{introduction-new-TBP}
	\input{relatedwork}
	\input{formulation}

	\input{background}
	\input{unifiedSYCOS} 
	\input{noisetheory}

	\input{distributedSYCOS}
	\input{experiment}

	\input{conclusion}
	
	\IEEEdisplaynontitleabstractindextext

	%
	\IEEEpeerreviewmaketitle

	\ifCLASSOPTIONcompsoc
	
	
	
	\bibliographystyle{abbrv}
	\bibliography{references} 
	
	
	
	
	
	\ifCLASSOPTIONcaptionsoff
	\newpage
	\fi

\end{document}

%% file: introduction-new-TBP.tex
\section{Introduction}
The massive spread of cheap IoT sensors has enabled the collection of big time series at unprecedented volume and velocity from different domains such as weather, energy, and transportation. Mining such big time series can help extract insights and unknown relationships from cross-domain datasets. One of the most fundamental analysis is to find correlations between time series, as the knowledge of how data are correlated with each other can enable more advanced analyses to support evidence-based decision making.  
For example, finding correlations between weather and transportation data can show how weather events like storm and rain affect traffic speed and use of taxis, from which public transport can be better scheduled; or knowing correlations between weather and energy production time series can help improve renewable energy forecasts. 
Furthermore, as correlation is one of the three building blocks to establish causal relations \cite{agresti2012social}, finding correlations can thus help construct causal inference models. 
           
Searching for correlations in big time series, however, is challenging. Specifically, an effective correlation search approach needs to deal with {\em (1) multiple temporal scales, (2) complex relations,} and {\em (3) noise}.
First, correlations often appear at multiple temporal scales, for instance, correlations involved weather data might range from hours (e.g., during rain showers), to days or weeks (e.g., during a storm), depending on the type of weather events.
Second, real-world time series exhibit different types of complex relations, including linear and non-linear, monotonic and non-monotonic, positive and negative, functional and non-functional. 
For example, stock prices or weather information exhibit non-linear relations, which rule out the use of traditional linear correlation metrics such as Pearson coefficients \cite{kenney2013mathematics}.
Lastly, big time series from IoT sensors are inherently noisy due to errors introduced by sensors malfunction, transmission, collection and processing. 
Although correlation search has been studied extensively in the literature, current state-of-the-art approaches are limited either to a single temporal scale, simple types of relations, e.g., only linear ones, and/or in their inability to handle noise (see Section 2 for details). Moreover, it is a common misconception that correlations and similarities in time series are the same concept, but they are indeed different since correlations capture the dependencies between data distributions, while time series similarities capture the similar values of data points identified by distance measure such as Euclidean distance. For example, if values in one time series go up, while they go down in another, the two time series are not similar, instead they have a {\em negative correlation}. Therefore, methods used for time series similarities search such as Dynamic Time Warping \cite{salvador2007toward} are not suitable to identify correlations.

\textbf{Contributions.} The present paper addresses all three challenges above by presenting the unified framework iSYCOS (integrated SYnchronous COrrelation Search) that supports: (1) multiple temporal scales, (2) complex relations, and (3) handling noisy data. Specifically, we make the following contributions. First, iSYCOS searches for correlations by relying on the {\em mutual information} (MI) between time series windows. The metric of mutual information is robust to noise and underlying data distributions, and is able to capture complex correlation relations, including linear and non-linear, monotonic and non-monotonic, positive and negative, functional and non-functional. Mutual information can also be computed for windows of any size, making it adaptive to multiple temporal scales.
Second, by combining mutual information as the metric and a window-based approach, iSYCOS can automatically discover correlations at multiple temporal scales without requiring users to specify the time scales or window sizes. Window sizes are automatically adjusted based on mutual information value which is used as an indicator for a potential correlation.
Third, we propose two window-based search methods, a top-down and a bottom-up one, to find correlations in time series datasets. The top-down approach, SYCOS$_{\text{TD}}$, starts with the largest window size and gradually narrows it down, and is thus suitable for time series containing correlations across large time scales. Reversely, the bottom-up approach, SYCOS$_{\text{BU}}$, starts with the smallest window size and gradually increases it, and is thus suitable for datasets where correlations are found at small time scales. 
Fourth, we propose a unified approach, iSYCOS, which integrates SYCOS$_{\text{TD}}$ and SYCOS$_{\text{BU}}$, and automatically selects the most appropriate method for a given dataset. Fifth, based on probability theory, we propose a novel theory to identify and disregard noise when computing mutual information, hence enabling efficient pruning of non-interesting time intervals from the search and significantly improving the search speed. This also improves the robustness of iSYCOS in processing noisy data. Sixth, iSYCOS proposes effective optimizations to reuse mutual information computations across neighboring windows, thus improving MI computation speed and making iSYCOS fast and memory efficient. Seventh, we propose a distributed version of iSYCOS based on Apache Spark which is scalable to large data volumes. Eighth, we perform time and space complexity analyses for both SYCOS$_{\text{TD}}$ and SYCOS$_{\text{BU}}$.
And finally, we perform a comprehensive experimental evaluation using synthetic and real-world datasets from the energy and smart city domains, showing that iSYCOS can automatically select the best approach (top-down or bottom-up), is able to find complex correlations at multiple time scales, that the noise pruning and optimizations are effective and have high accuracy, and that iSYCOS can scale out linearly on a cluster. 
	
\textit{Paper Outline.} The paper is organized as follows. Section \ref{sec:relatedwork} discusses the related work. 
Section~\ref{sec:formulation} formulates the window-based synchronous correlation search problem.
Section~\ref{sec:background} presents the background of mutual information and MI-based correlation.
Section~\ref{sec:method} presents SYCOS$_{\text{TD}}$, SYCOS$_{\text{BU}}$, and the integrated iSYCOS. Sections \ref{sec:noisetheory} and \ref{sec:distributedSYCOS} discuss the noise pruning, the optimized MI computation, and the distributed iSYCOS. Section \ref{sec:experiment} presents the experimental evaluation and finally, Section~\ref{sec:conclusion} concludes and points to future work.

%% file: relatedwork.tex
\section{Related Work}\label{sec:relatedwork}
Detecting correlations in time series has been extensively studied in the literature. Most of the proposed methods use traditional statistical metrics such as covariance and correlation coefficients (e.g., Pearson, Spearman) to measure correlations \cite{dean1974application,huang2010application,zhang2003correlation,ho2017towards,ho2015data,cappiello2015co,barkat2014open,ho2015data,ho2016characterizing,gribaudo2014analysis,ho2017improving,ho2013activity,ho2017towards}. However, these metrics can only identify linear and monotonic relations. 
To capture more complex dependencies such as non-linear and non-monotonic, recent work \cite{das2012finding, pochampally2014fusing,  alawini2014helping, roy2014formal, yang2002delta,de2007fast, middelfart2013efficient,chirigati2016data,xie2013local} attempts to address the problem from a higher level. For example, in \cite{das2012finding}, Sarma et al. use \textit{relatedness} to capture different kinds of related-relations, e.g., entity/schema complement, between data tables. 
Pochampally et al. \cite{pochampally2014fusing} use \textit{joint precision} and \textit{joint recall} to measure the similarity between datasets. Alawini et al. \cite{alawini2014helping} rely on the history and schema of datasets to map and link them. Roy et al. \cite{roy2014formal} use the concept of \textit{intervention} to capture query relations. Sousa et al. \cite{de2007fast} use the \textit{intrinsic dimension} to detect correlated attributes in databases. Cai et al. in \cite{cai2018effective} use \textit{cohort} analysis to look for causal explanations in the data. Yang et al. \cite{yang2002delta} use a \textit{residue} metric that measures the difference between the actual and expected value of an object to capture object correlations in large datasets. 
Middelfart et al. \cite{middelfart2013efficient} propose a bitmap-based approach to discover \textit{schema-level} relationships in a multidimensional data cube.
Chirigati et al. \cite{chirigati2016data} propose a topology-based framework to identify relationships between spatio-temporal datasets during both regular and extreme events. 
Other work such as \cite{xie2013local} uses Discrete Fourier Transform to detect local correlations in streaming data, with the main focus on the linear correlations.

Concerning correlation discovery techniques, the window-based technique has been used in \cite{schulz2002sliding,cole2005fast,keller2015estimating,ho2020efficientpattern,ho2020efficientpatternfull,ho2021efficientdistributed}. However, they adopt different correlation metrics, e.g., Schulz et al. \cite{schulz2002sliding} use the Spearman coefficient, while Cole et al. \cite{cole2005fast} use sketches. Other work such as \cite{camerra2014beyond} focuses on techniques for indexing and representing the time series, rather than the correlation search method itself.

Different from all the above work, our proposed iSYCOS framework uses MI as the correlation measure, and thus can discover not only linear and monotonic relations, but also non-linear and non-monotonic ones. 
By relying on a window-based technique to search for correlations, iSYCOS can find the exact time when the correlations occur, and thus provide richer information about the dependencies. Furthermore, we propose an MI-based theory that can identify noise in the data and thus, prune irrelevant data partitions during the search to improve its performance. We also optimize the MI computation to reuse information across windows, further speeding up the search. 

The MI measure has been broadly used in numerous domains to achieve different goals, e.g., feature selection \cite{estevez2009normalized}, clustering and mining \cite{ke2008information}, image alignment and registration \cite{pluim2003mutual}, and network inference \cite{meyer2007information}. However, using MI to find correlations in a Big Data context is still a new area. Su et al. \cite{su2014supporting} use MI to analyze relationships of massive scientific datasets. However, they do not consider window-based correlations, but instead only focus on the overall correlations of variables. Keller et al.  \cite{keller2015estimating} propose the MISE algorithm to estimate MI for streaming data, but did not use this MI estimation to infer correlation. 

In the present paper, the proposed iSYCOS framework is built upon the earlier versions of this work which investigated the two underlying approaches, bottom-up \cite{ho2020efficient}, \cite{ho2019efficient} and top-down \cite{ho2016adaptive}, \cite{ho2019amic}. In \cite{ho2020efficient} and \cite{ho2019efficient}, we investigated how the Late Acceptance Hill Climbing algorithm \cite{burke2017late} can be used to search for window-based correlations in a bottom-up fashion. 
In \cite{ho2016adaptive}, \cite{ho2019amic}, we proposed the top-down approach that uses a  hierarchical multi-layer search to discover correlated windows. In the present paper, we extend the previous work in the following ways. First, we unify the two approaches, top-down SYCOS$_{\text{TD}}$ and bottom-up SYCOS$_{\text{BU}}$, into one single framework iSYCOS. Second, we propose a set of efficiency scores to assess the performance of SYCOS$_{\text{TD}}$ and SYCOS$_{\text{BU}}$, and a selection algorithm using the efficiency scores to automatically analyze the input time series and select the most suitable method among SYCOS$_{\text{BU}}$ and SYCOS$_{\text{TD}}$ for a given input. Third, we integrate the novel MI-based noise theory, introduced previously in the bottom-up approach \cite{ho2020efficient}, into the top-down SYCOS$_{\text{TD}}$ to reduce the search space and improve its performance. 
Fourth, we develop a distributed version of iSYCOS using Apache Spark, allowing SYCOS$_{\text{BU}}$ and SYCOS$_{\text{TD}}$ to scale to big datasets on a cluster. Fifth, we provide the time and space complexity analyses for both SYCOS$_{\text{BU}}$ and SYCOS$_{\text{TD}}$, and conduct extensive experiments to compare their performance, showing which method is best for which specific scenarios. 

%% file: formulation.tex
\section{Problem Formulation}\label{sec:formulation}
We now define the window-based correlations, and formulate the Synchronous Correlation Search (SYCOS) problem.

\textit{Definition 3.1} (\textit{Time series}) A time series $X_T=\{x_1,x_2,...,x_n\}$ is a sequence of data values that measures the same phenomenon during an observation time period $T$, and is chronologically ordered. 

The observation time period $T=[t_1,t_n]$ contains $n$ time steps $t_i$, each recording a data value $x_{i} \in X_T$, where $t_1$ and $t_n$ denote the first and the last time step of $T$, respectively. We say $X_T$ has length $n$ if $X_T$ contains $n$ data values. 

\textit{Definition 3.2} (\textit{Time window}) A \textit{time window} $w=[t_{s},t_{e}]$ is a temporal sub-interval of $T$ that records the events of $X_T$ from the start time $t_{s}$ to the end time $t_{e}$, and forms a (sub) time series $X_w=\{x_{t_s},...,x_{t_e}\} \in X_T$. 
We say that $w$ has size $m$, denoted as $\abs{w}=m$, if $X_w$ contains $m$ data values.

\textit{Definition 3.3} (\textit{Pair of time series}) A pair of two time series $(X_T,Y_T)=(\{x_1,x_2,...,x_n\},\{y_1,y_2,...,y_n\})$ contains data collected from $X_T$ and $Y_T$ that measure two separate phenomena during the same time interval $T$. A tuple $(x_i,y_i) \in (X_T,Y_T)$ records the data values of $X_T$ and $Y_T$ at the same time instance $t_i$.  

\textit{Definition 3.4} (\textit{Pair of time windows}) A pair of time windows $(w_X, w_Y)=([t_{x_s},t_{x_e}],$ $[t_{y_s},t_{y_e}])$ is a pair of temporal sub-intervals of $T$ that records the events of $X_T$ and $Y_T$ during $[t_{x_s},t_{x_e}]$ and $[t_{y_s},t_{y_e}]$, respectively. 
The pair $(w_X, w_Y)$ is \textit{synchronous} iff $t_{x_s}=t_{y_s} \wedge t_{x_e}=t_{y_e}$. For simplicity, the \textit{synchronous} pair $(w_X, w_Y)$ is called a \textit{time window} of $(X_T,Y_T)$, and denoted as $w_{X,Y}=[t_s,t_e]$ where $t_s$ and $t_e$ are the start time and end time of the pair, and $t_s=t_{x_s}=t_{y_s} \wedge t_e=t_{x_e}=t_{y_e}$.  

\textit{Definition 3.5} (\textit{Correlated time window}) Consider a pair of time series $(X_T,Y_T)$  measured during the time interval $T$, and a \textit{time window} $w_{X,Y}$ of $(X_T,Y_T)$. Let $f(.)$ be a predefined correlation function, and $\sigma$ be a predefined correlation threshold. The two time series $X_T$ and $Y_T$ are said to be \textit{correlated within the window} $w_{X,Y}$ iff $X_w$ and $Y_w$ are correlated according to $f(.)$, i.e., $f(X_w,Y_w) \ge \sigma$. 

The function $f(.)$ can be any function that measures  data dependencies, for example, Pearson correlation coefficient, Spearman correlation, Euclidean distance, or Mutual Information. Based on the above definitions, we formulate the synchronous correlation search problem as follows.

\textit{\bf SYnchronous COrrelation Search (SYCOS)}: Let $(X_T,Y_T)$ be a pair of time series measured during the time interval $T$, $s_{\min}$ and $s_{\max}$ be the minimum and maximum sizes that a time window $w_{X,Y}$ can have, $f(.)$ be a correlation function, and $\sigma$ be the predefined correlation threshold. 
Then SYCOS finds a set $S$ of all non-overlapping time windows $w_{X,Y}=[t_s,t_e]$ such that $t_1 \le t_s < t_e \le t_n$ $\wedge$ $s_{\min} \le \abs{w_{X,Y}} \le s_{\max}$ $\wedge$ $f(X_w,Y_w) \ge \sigma$ $\wedge$ $\forall w_i, w_j \in S$$:$ $ w_i \cap w_j = \emptyset$.

The \textit{SYCOS} problem aims to find all non-overlapping time windows of a given time series pair such that their sub time series are correlated according to the correlation function and the predefined threshold.

%% file: background.tex
\section{Mutual Information-Based Correlation}\label{sec:background}
In Section \ref{sec:formulation}, we use a general correlation function $f(.)$ to define correlated time windows. In this section, we discuss our choice of using mutual information as a correlation measure, and argue why it is a good fit for the correlation search problem.
\subsection{Mutual Information}
Mutual information is a statistical measure to quantify the shared information between two probability distributions. Given two discrete random variables $X$, $Y$ with the corresponding probability mass functions (pmfs) $p(x)$, $p(y)$, and the joint distribution $p(x,y)$, the mutual information between $X$ and $Y$ is defined as 
\begin{equation}
\begin{split}
I(X;Y)=\sum_{y \in Y} \sum_{x \in X} p(x,y) \log\frac{p(x,y)}{p(x)p(y)}
\end{split}
\label{eq:mi}
\end{equation}

Intuitively, $I(X;Y)$ represents the reduction of uncertainty of one variable (e.g., $X$) given the knowledge of another variable (e.g., $Y$) \cite{cover2012elements}. The larger $I(X;Y)$, the more information is shared between $X$ and $Y$, and thus, the less uncertainty about one variable when knowing the other. Mutual information possesses multiple appealing properties that lead to its widespread use in many different fields, including machine learning, neuroscience, and data compression \cite{ho2017towards}. Given the mutual information $I(X;Y)$, the following properties hold: (1) Non-negativity: $I(X;Y) \ge 0$; (2) Invariance under one-to-one transformations: $I(X;Y) = I(U(X);V(Y))$ where $U$ and $V$ are homeomorphisms that create homeomorphic mappings on $X$ and $Y$; (3) Chain rule expansion: $I(X_1,...,X_n;Y) = \sum_{i=1}^{n} I(X_i;Y|X_{i-1},...,X_1)$.

The MI value is equal to zero if and only if the considered variables are statistically independent, and otherwise always positive if there exists any kind of dependency (e.g., functional and non-functional, linear and non-linear) \cite{de2013comparative}. This property makes MI a versatile measure to capture correlations in noisy data sets which often exhibit a high degree of bias and abnormality, causing their relationships to often be arbitrary and non-linear. 

\subsection{Estimating Mutual Information} Although MI is an important measure of information for many applications, applying MI in practice is challenging due to the difficulty in estimating probability distributions. As finding a more efficient method to estimate MI from collected samples is still an active research problem, in this work, we adopt the MI estimator proposed by Kraskov et al. \cite{kraskov2004estimating} (hereafter called \textit{KSG}) for several reasons: (1) \textit{KSG} outperforms other estimators (e.g., histogram, kernel density estimation) in terms of computational efficiency and accuracy \cite{papana2009evaluation}; (2) \textit{KSG} uses $k-$nearest neighbor approximation and thus is data efficient (i.e., it does not require very large sets of samples), is adaptive and has minimal bias \cite{kraskov2004estimating}. 

The main idea of \textit{KSG} is that instead of directly computing the joint and marginal probability distributions of the considered variables, it approximates the distributions by computing the densities of data points in nearby neighborhoods \cite{kraskov2004estimating}. Specifically, \textit{KSG} computes the probability distribution for the distance between each data point and its $k^{\text{th}}$ nearest neighbor. For each data point, it first searches for $k$ nearest neighbor clusters ($k$ is a pre-defined parameter) and computes the distance $d_k$ to the $k^{\text{th}}$-neighbor. Then, the population density within the distance $d_k$ is estimated by counting the number of data points that fall inside $d_k$. 
This leads to the computation of MI between $X$ and $Y$ as \cite{kraskov2004estimating}: 
\begin{equation}
I(X;Y) \approx \psi(k) - 1/k  - \left<\psi(n_{x}) + \psi(n_{y})\right> + \psi(n)
\label{formula:KSGmi}
\end{equation}
where $\psi$ is the digamma function, $k$ is the number of nearest neighbors, $(n_{x}, n_{y})$ is the number of marginal data points in each dimension falling within the distance $d_k$, $n$ is the total number of data points and $\left<\cdot\right>$ is the average function. 
The digamma function $\psi$ is a monotonically increasing function. Thus, the larger $n_x$ and $n_y$ (i.e., more data points fall within the distance $d_k$), the lower $I(X;Y)$, and vice versa. In practice, mutual information can be computed using empirical distributions as follows. 

\textit{Definition 4.1} (\textit{Empirical distribution}) Let $X_T$ be a time series measured during a time interval $T$. A distribution $\hat{p}$ is an empirical distribution of $X_T$, denoted as $\hat{p}(x_i)$, if $\hat{p}$ is estimated from the empirical measured values $x_i \in X_T$, $i=1,...,n$. 

\textit{Definition 4.2} (\textit{Mutual information of a time window}) Let $(X_T,Y_T)$ be a pair of time series measured during the time interval $T$, and $w_{X,Y}$ be a \textit{time window} of $(X_T,Y_T)$. The mutual information between $X$ and $Y$ within $w_{X,Y}$ can be estimated using the empirical distributions $\hat{p}$ of data within this window as
\begin{equation}
I_{w_{X,Y}}=I(X_w;Y_w) \approx \sum_{y_i \in Y_w} \sum_{x_i \in X_w} \hat{p}(x_i,y_i)\log\frac{\hat{p}(x_i,y_i)}{\hat{p}(x_i)\hat{p}(y_i)}
\end{equation}


%% file: unifiedSYCOS.tex
\section{The Unified  SYCOS Search Framework}\label{sec:method}
We now present our unified framework iSYCOS for synchronous correlation search. We start by describing different scenarios where correlated time windows are distributed differently on a given pair of time series, and how different search approaches are beneficial for each of those scenarios. 

\begin{figure*}[t]
	\vspace{-0.05in}
	\centering
	\begin{subfigure}{0.33\textwidth}
		\centering
		\includegraphics[width=\textwidth]{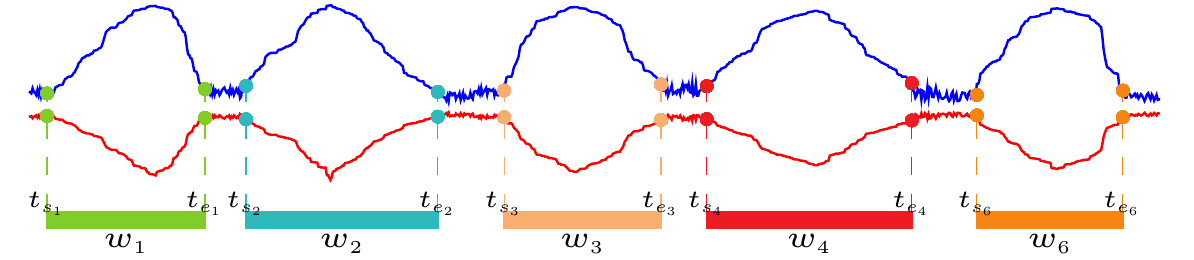} 
		\vspace{-0.2in}
		\caption{Dense correlated time windows}
		\label{fig:top_down_scenario}
	\end{subfigure}
	\begin{subfigure}{0.33\textwidth}
		\centering
		\includegraphics[width=\textwidth]{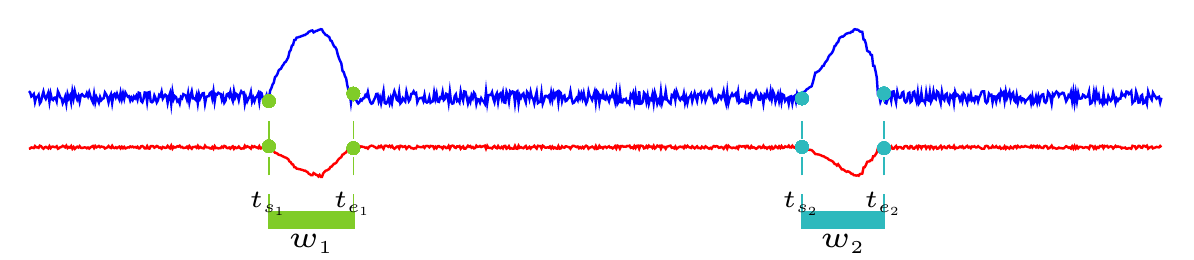} 
		\vspace{-0.2in}
		\caption{Sparse correlated time windows}
		\label{fig:bottom_up_scenario}
	\end{subfigure}	
	\begin{subfigure}{0.33\textwidth}
		\centering
		\includegraphics[width=\textwidth]{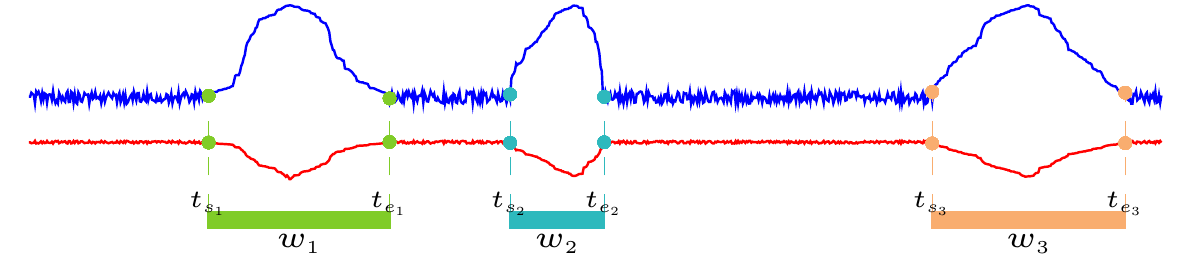} 
		\vspace{-0.2in}
		\caption{Moderate correlated time windows}
		\label{fig:selection_scenario}
	\end{subfigure}
	\vspace{-0.1in}
	\caption{Different correlation search scenarios}
	\label{fig:scenarios}
		\vspace{-0.2in}
\end{figure*}

Figure \ref{fig:scenarios} outlines three scenarios of correlated time windows in the time series. The first scenario in Figure \ref{fig:top_down_scenario} illustrates a situation where correlated time windows are large and located close or next to each other. In such a scenario, a search strategy that starts with a large window, and only reduces the window size when needed, is more efficient to identify correlated windows. We call this search strategy a \textit{top-down} approach. The second scenario in Figure \ref{fig:bottom_up_scenario} is when correlated windows are small and sparsely distributed. In this scenario, a search strategy that starts from a small window, and incrementally extends the window size only when needed, is faster to locate correlated windows. We call this search strategy a \textit{bottom-up} approach. 
Finally, the third scenario in Figure \ref{fig:selection_scenario} is when correlated windows are both small and large, and are neither close nor far away from each other. In this scenario, we need a strategy to select between bottom-up and top-down search. 

The following sections will introduce the top-down and bottom-up approaches separately. We then present a selection algorithm that automatically evaluates and selects the most efficient search approach between top-down and bottom-up for a given time series pair.

\subsection{SYCOS$_{\text{TD}}$: Top-Down Search}\label{sec:topdown}
The top-down approach SYCOS$_{\text{TD}}$ uses a breadth-first search strategy to search for correlated time windows: it starts with a large window, and reduces the size when needed as the search progresses. Intuitively, SYCOS$_{\text{TD}}$ assumes that correlations exist in large windows, and thus, it is more efficient to start with those. 

\subsubsection{Multi-layer hierarchical search space}
To perform the top-down search, SYCOS$_{\text{TD}}$ divides the search space into multiple hierarchical layers, each corresponding to one specific window size. The \textit{temporal scale} is used to keep track of window sizes among different layers. 

\textit{Definition 5.1} (\textit{Temporal scale}) 
A \textit{temporal scale} represents the size of windows belonging to the same layer in a multi-layer hierarchical search space.

Examples of temporal scale are year, month, week, day, and so on. Users can define the temporal scale using their domain knowledge. For instance, windows with temporal scale \textit{year} contain time series of one year, while windows of \textit{day} temporal scale contain time series of one day. 
Fig. \ref{fig:topdownsearchtree} depicts the multi-layer hierarchical search space. Here, the top layer L$_1$ has the coarsest temporal scale, corresponding to the largest windows. The bottom layer L$_n$ has the finest temporal scale, corresponding to the smallest windows. 

\subsubsection{Top-down search with sliding window}
Given the multi-layer search space, initially, SYCOS$_{\text{TD}}$ starts out at the top layer L$_1$, using the largest window size $s_{\max}$. In each layer, SYCOS$_{\text{TD}}$ uses a sliding window to slide over the time series, identifying correlated windows and filtering out uncorrelated ones. 
Specifically, to determine whether the time series in a time window $w_i$ are correlated, the MI $I_{w_i}$ is computed and compared against the threshold $\sigma$: $w_i$ is a \textit{correlated} window if $I_{w_i} \ge \sigma$, and \textit{uncorrelated} otherwise. The result set $S$ is used to collect correlated windows, while uncorrelated ones are passed down to the next lower layer of the search space. As the search progresses, SYCOS$_{\text{TD}}$ traverses to lower layers until it reaches the lowest one L$_n$ which has the smallest window size $s_{\min}$. 

In the same layer, two consecutive windows $w_i$ and $w_{i+1}$ can be either disjoint or overlapping depending on whether significant correlation exists in one of them. More specifically, if $w_i$ is a correlated window, then $w_i$ and $w_{i+1}$ are disjoint. In contrast, if $w_i$ is an uncorrelated window, then $w_{i+1}$ is created by shifting $w_i$ by a predefined $\delta$-step, making $w_i$ and $w_{i+1}$ overlap each other. 
We illustrate this sliding window technique in Fig. \ref{fig:topdownsearchtree}. 

SYCOS$_{\text{TD}}$ starts out at L$_1$ with the first window $w_0=[t_{s_0},t_{e_0}]$ of size $s_{\max}$. The MI value $I_{w_0}$ is computed and compared against $\sigma$ to determine whether $w_0$ is a correlated window. 
In this example, we assume $I_{w_0}<\sigma$, and thus, $w_{0}$ is uncorrelated. SYCOS$_{\text{TD}}$ passes $w_0$ down to layer L$_2$ and moves to the next window $w_1 \in L_1$. 
To define $w_1$, SYCOS$_{\text{TD}}$ shifts $w_0$ by a $\delta$-step, forming $w_1=[t_{s_1},t_{e_1}]$ as $t_{s_1}=t_{s_0}+\delta$ $\wedge$ $t_{e_1}=t_{e_0}+\delta$. 
The MI of $w_1$ is then evaluated, which results in $I_{w_1} < \sigma$. Here, $w_0$ and $w_1$ are two consecutive and uncorrelated windows, thus, their indices are merged into one single \textit{uncorrelated data partition} $[t_{s_{0}}, t_{e_{1}}]$. This \textit{uncorrelated data partition} will be used by SYCOS$_{\text{TD}}$ in the next search iteration in the lower layer. 

Next, SYCOS$_{\text{TD}}$ moves to $w_2$ which is formed by shifting $w_1$ a $\delta$-step. Assume that at $w_2$, we have $I_{w_2} \ge \sigma$, hence, the time series in $w_2$ are correlated. The window $w_2$ is inserted into the result set $S$, after which the \textit{uncorrelated data partition} $[t_{s_{0}}, t_{e_{1}}]$ previously formed by $w_0$ and $w_1$ is updated to $[t_{s_0},t_{s_2}]$, indicating that only the time series from $t_{s_0}$ to $t_{s_2}$ are uncorrelated. The search moves on to the next window $w_3=[t_{s_3},t_{e_3}]$. Here, since $w_2$ is a correlated window, $w_3$ starts right after $w_2$ with $t_{s_3}=t_{e_2}+1$. The procedure is repeated until the entire time series in L$_1$ is scanned.

At L$_2$ and lower layers, SYCOS$_{\text{TD}}$ repeats the search procedure on \textit{uncorrelated data partitions} passed down from the above layers, using progressively smaller windows until the lowest layer $L_n$ is scanned. 

\begin{figure*}[!t]
	\centering
	\begin{minipage}{0.49\linewidth}
		\vspace{-0.1in}
		\hspace{-0.1in}
		\centering
		\includegraphics[scale=0.85]{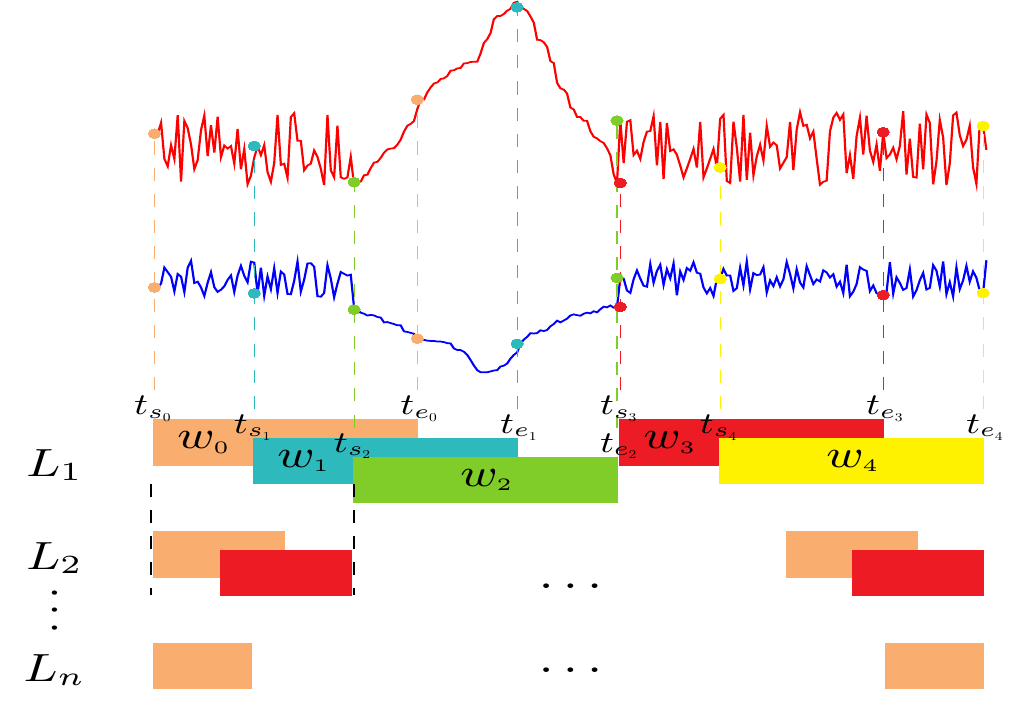} 
		\vspace{-0.1in}
		\caption{SYCOS$_{\text{TD}}$: multi-layer hierarchical search tree}
		\label{fig:topdownsearchtree}
	\end{minipage}
	\vspace{0.1in}
	\begin{minipage}{0.49\linewidth}
		\centering
		\hspace{0.1in}
		\includegraphics[scale=1.7]{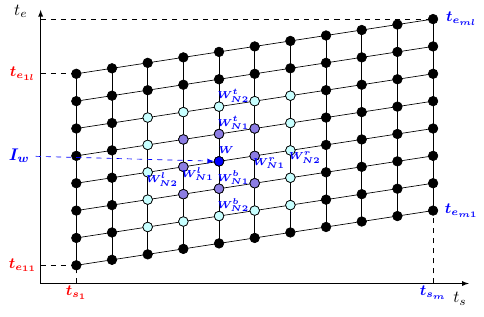} 
		\vspace{-0.1in}
		\caption{SYCOS$_{\text{BU}}$ search space}
		\label{fig:searchspacebf}
	\end{minipage}	
	\vspace{-0.1in}
\end{figure*}

Algorithm \ref{alg:topdownsearch} gives the pseudocode for SYCOS$_{\text{TD}}$. In line 1, SYCOS$_{\text{TD}}$ starts at L$_1$ with windows of size $s_{\max}$. Lines 6-14 perform sliding window search on each layer. Line 17 reduces the temporal scale in order to move to the next lower layer. SYCOS$_{\text{TD}}$ stops when it finishes the search of the lowest layer L$_n$ (line 2).

\begin{algorithm}
	\caption{SYCOS$_{\text{TD}}$: Top-Down Search}
	\footnotesize
	\begin{algorithmic}[1]
		\Statex{{\bf Input:}} $\{X_T,Y_T\}$: \text{pair of time series} 
		\Statex{{\bf Params:} $\sigma$: \text{correlation threshold,}} 
		\Statex{\hspace{1.07cm}$\delta$: \text{sliding step,}} 
		\Statex{\hspace{1.07cm}$s_{\min}$, $s_{\max}$: \text{minimum and maximum window sizes}} 
		\Statex{{\bf Output:} $S$: \text{set of non-overlapping windows with MI $\ge \sigma$}}
		
		\State{\textit{Initialize:} $s_w \gets s_{\max}$, $\textit{uncorrelatedPartitions} \gets \{X_T,Y_T\}$}
		\While{$s_{w} \ge s_{\min}$}
		\While{notEmpty({$\textit{uncorrelatedPartitions}$})}
		\State {$\textit{uPartition}$ $\gets$ {$\textit{uncorrelatedPartitions.next}$}} \Comment{gets the next uncorrelated partition}
		\State \text{$w$ $\gets$ getNextWindow({$\textit{uPartition}$}$,s_{w})$} \Comment{gets the next window from the current uncorrelated partition with size $s_w$}
		\While{\text{$\textit{w.endIndex}$ $\le$ $\textit{uPartition.endIndex}$}} 
		\State	 \text{$I_w$ $\gets$ computeMI($w$)} \Comment{compute MI value}
		\If{\text{$I_w \ge \sigma$}}
		\State \text{$\mathit{S}$.insert($w$)} 
		\State {$w$ $\gets$ getNextWindow($\textit{uPartition},s_w)$}
		\Else
		\State \text{$w$ $\gets$ shiftWindow($\textit{uPartition},s_w,\delta)$} 
		\EndIf
		\State \text{update($\textit{uncorrelatedPartitions}$)}
		\EndWhile
		\EndWhile
		\State $s_w  \gets \text{reduce}(s_w)$ \Comment{reduce the temporal scale in lower layer}
		\EndWhile
		\State \text{\textbf{return} $\mathit{S}$}
	\end{algorithmic}
	\label{alg:topdownsearch}
\end{algorithm}

\subsection{SYCOS$_{\text{BU}}$: Bottom-Up  Search}\label{sec:bottomup}
In contrast to the top-down approach, the bottom-up approach SYCOS$_{\text{BU}}$ starts with a small window, and incrementally enlarges the window size as the search progresses. SYCOS$_{\text{BU}}$ assumes that correlations exist in small windows, and thus, it is more efficient to start with those.
\subsubsection{Late Acceptance Hill Climbing}
SYCOS$_{\text{BU}}$ employs Late Acceptance Hill Climbing (LAHC) \cite{burke2017late}, an extension of the classic Hill Climbing (HC) \cite{russell2016artificial}, to navigate through the search space while looking for correlated windows. 
Particularly, given a \textit{target function} $f$ and a current solution $S$ of $f$, LAHC tries to improve $S$ by exploring the neighborhood of $S$ to find new solutions that are better than $S$. If a better solution is found (according to the predefined objectives), the current solution $S$ is replaced by this new solution $S_{\text{new}}$, and the process is repeated until no further improvement can be made.   
Compared to classic HC, LAHC is different in terms of its acceptance rule, i.e., it can accept a solution $S_{\text{new}}$ if $S_{\text{new}}$ is better than either the current solution $S$ or a solution $S_{\text{old}}$ found in the history. To do that, LAHC uses a fixed length array $L_h$ to maintain a history of the most recently accepted solutions, and use this $L_h$ array to justify the goodness of a candidate solution. 


\subsubsection{Applying LAHC to SYCOS$_{\text{BU}}$}
We consider the problem of searching for synchronous correlations using LAHC as a \textit{maximization} problem. Specifically, the \textit{target function} of LAHC is a \textit{maximize} function, and our goal is to search for the time windows where their MIs are locally maximal. To explain why local optima are in fact the solutions of SYCOS$_{\text{BU}}$, consider Fig. \ref{fig:mifluctuation} that plots the fluctuation of MI values for a given pair of variables. Here, the $y-$axis represents the MI values of the time windows on the $x-$axis. Given the correlation threshold $\sigma$ (red line in Fig. \ref{fig:mifluctuation}), the three windows which correspond to the three maximal points (in red) indicate highly correlated areas (above $\sigma$). 
These windows can be found by first identifying the three peak points which represent the locally maximal MI values, and then extending the located points to the surrounding areas to determine the exact indices of the windows. Such windows are guaranteed to satisfy the correlation condition $I_w \ge \sigma$. Since an LAHC-based method guarantees local optimal solutions, it is an ideal technique for SYCOS. 

\hspace{-0.2in}\textit{Definition 5.2} \textit{($\delta$-neighbor)} Let $(X_T,Y_T)$ be a pair of time series of length $n$ measured during the time interval $T$, and $w=[t_s,t_e]$ be a time window of $(X_T,Y_T)$. A window $w^{'}=[t_s^{'},t_e^{'}]$ is a \textit{$\delta$-neighbor} of $w$ if its start time $t_s^{'}$ and/or its end time $t_e^{'}$ differs by a $\delta$ value from $t_s$ and $t_e$, respectively. i.e., $t_s^{'}=t_s + \delta$ $\vee$ $t_s^{'}=t_s - \delta$ $\vee$ $t_e^{'}=t_e + \delta$ $\vee$ $t_e^{'}=t_e - \delta$, where $\delta$ is a pre-defined moving step such that $1 \le \delta \le n$ $\wedge$ $s_{\min} \le \mid$${w^{'}}$$\mid \le s_{\max}$. 
\\
\textit{Definition 5.3} \textit{($\delta$-neighborhood)} Let $w=[t_s,t_e]$ be a time window of $(X_T,Y_T)$. A \textit{$\delta$-neighborhood} $N$ of $w$, denoted as $N_{\delta}$, is the surrounding area of $w$ formed by all \textit{$\delta$-neighbors} $w^{'}=[t_s^{'},t_e^{'}]$ of $w$. 

We illustrate the \textit{neighborhood} concept in Fig. \ref{fig:searchspacebf}. 
Here, each point in the grid represents a time window where the x-axis indicates the start time, and the y-axis indicates the end time. Consider the window $w$ in blue. The nearest $\delta-$neighborhood of $w$, called the $1-$neighborhood $N_1$, is the area formed by the \textit{eight} windows in light purple color surrounding $w$. Each window in this neighborhood differs from $w$ by \textit{one} $\delta$ step, either by its start index, or its end index, or both. Going further from $w$, the next neighborhood of $w$, called the $2-$neighborhood $N_2$, is the $16$ windows in light blue color. Each $\delta-$\textit{neighborhood} forms an area where LAHC will iteratively look for solution candidates of a given window $w$. 
We provide the pseudocode of SYCOS$_{\text{BU}}$ in Algorithm \ref{alg:bottomupsearch}, and explain the details below.

Given a pair of time series $(X_T,Y_T)$, and the \textit{target function} $I_w$ to be maximized, SYCOS$_{\text{BU}}$ starts with an initial solution $w=w_0$ where $\abs{w_0}=s_{\min}$ (Algorithm \ref{alg:bottomupsearch}, line $2$), and evaluates the goodness of $w$ by computing the MI $I(w_0)$ (line 3). To maximize $I_w$, SYCOS$_{\text{BU}}$ first searches for a better solution in its nearest neighborhood $N_1$. It does so by computing $I(w^{'})$ for each neighbor window $w^{'} \in N_1$, and selects the best neighbor $w_{\text{best}}$  that has the highest MI in the neighborhood (lines 7-11). Then, SYCOS$_{\text{BU}}$ determines whether $w_{\text{best}}$ is better than the current solution $w$ using the following policies:
\\
$\bullet$ P1: If $I_{w_{\text{best}}} \ge I_w$ or $I_{w_{\text{best}}} \ge I_{w_h}$ where $w_h \in L_h$, then $w_{\text{best}}$ is better than $w$ and thus, $w$ is replaced by $w_{\text{best}}$ (lines 13-15).
\\
$\bullet$ P2: If $I_{w_{\text{best}}} < I_w$ and $I_{w_{\text{best}}} < I_{w_h}$ where $w_h \in L_h$, there is no better solution found in the considered neighborhood, thus, no improvement can be made (lines 16-18).

In policy P1, a better solution is found, thus, SYCOS$_{\text{BU}}$ moves to this new solution $w=w_{\text{best}}$, and repeats the exploration process on the new $w$. In policy \textit{P2}, no better solution is found. In this case,  SYCOS$_{\text{BU}}$ checks the \textit{stopping conditions} and determines whether it should continue exploring a larger neighborhood or stopping the search immediately. Ideally, SYCOS$_{\text{BU}}$ will stop immediately when no better solution can be found in the current neighborhood. However, to avoid situations where a temporary setback stops the search too early, we use an idle period to measure the number of non-improvements observed. The search will stop when it reaches the pre-defined max idle period $T_{\text{maxIdle}}$ (line 6).
The value $I_w$ at the stopping point is the locally maximal value of the target function, and is compared against the correlation threshold $\sigma$. The window $w$ is inserted into the result set $S$ if $I_w \ge \sigma$ (line 24). 

When SYCOS$_{\text{BU}}$ stops, the time series might not be scanned entirely. In that case, SYCOS$_{\text{BU}}$ restarts on the remaining part of the data, looking for new local optima, until the entire time series is scanned (line 25). 


\textit{Initial solution:} The first window $w_0$ of SYCOS$_{\text{BU}}$ can be at the beginning of the time series, or at an arbitrary position within the time series. 

\textit{The history list $L_h$:} SYCOS$_{\text{BU}}$ maintains a history list $L_h$ of the most recently accepted solutions and uses it to justify the goodness of a solution candidate. In our implementation, SYCOS$_{\text{BU}}$ follows the \textit{random} policy when selecting and updating an item in $L_h$: to justify a solution candidate, SYCOS$_{\text{BU}}$ randomly picks an element from the list and compares to the candidate; similarly, a random element from $L_h$ is selected for updating. This \textit{random} policy helps SYCOS$_{\text{BU}}$ escape \textit{plateaus} in the search space by creating "randomness" while it moves.

\begin{figure}[!t]
	\hspace{-0.1in}
	\begin{minipage}{0.4\linewidth}
		\centering
		\vspace{0.15in}
		\includegraphics[scale=0.145]{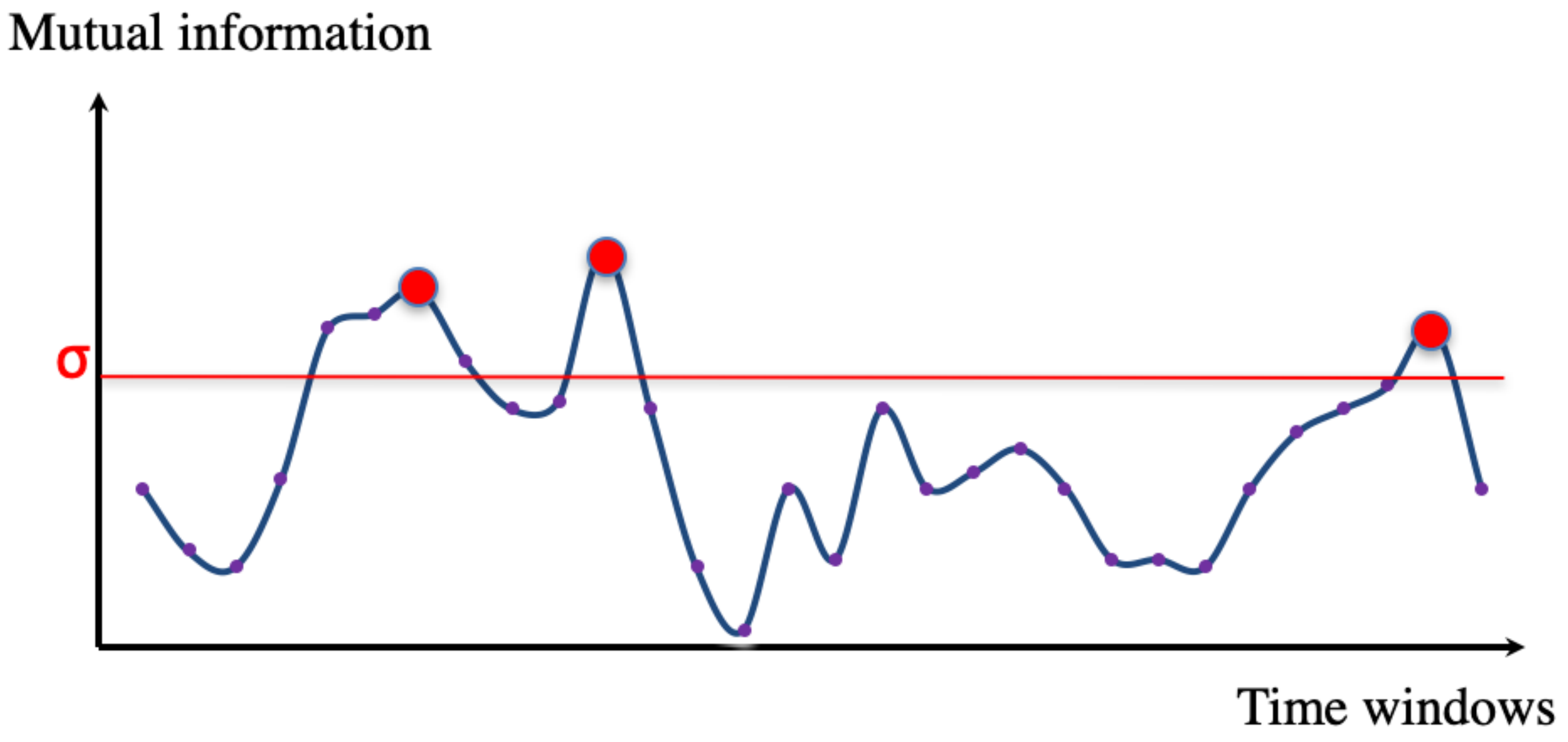}
		\caption{\small MI fluctuation}
		\label{fig:mifluctuation}
	\end{minipage}
	\hspace{0.3in}
	\begin{minipage}{0.6\linewidth}
		\hspace{0.15in}
		\begin{tikzpicture}[scale=0.45]
			\begin{axis}[
			xlabel=End Index, 
			ylabel=Mutual Information,
			ymin=0, ymax=0.3,
			]
			\addplot[blue, mark=none] table[x=end, y=mi] {mi_start0};
			\addplot[red, mark=none] table[x=end, y=mi] {mi_start5};
			\legend{Start at 0, Start at 5}
			\end{axis}
		\end{tikzpicture}
		\vspace{-0.1in}
		\caption{\small MI with different starts}
		\label{fig:mi3}
	\end{minipage}
\end{figure}

\begin{algorithm}[!t]
	\footnotesize
	\caption{SYCOS$_{\text{BU}}$: Bottom-Up Search}
	\begin{algorithmic}[1]
		\Statex{{\bf Input:}} $(X_T,Y_T)$: \text{pair of time series} 
		\Statex{{\bf Params:} $\sigma$: \text{correlation threshold,}}
		\Statex{\hspace{1.17cm}$\delta$: \text{moving step,}}
		\Statex{\hspace{1.17cm}$s_{\min}$, $s_{\max}$: \text{minimum and maximum window sizes}}  
		\Statex{{\bf Output:} $\mathit{S}$: \text{set of non-overlapping windows with MI $ \ge \sigma$}}
		
		\While{$(X_T,Y_T)$ are not completely explored}
		\State Produce an initial solution $w \gets w_0$ with $\abs{w_0}=s_{\min}$ 
		\State Compute $I(w_0)$ \Comment{Evaluate the goodness of $w_0$}
		\State For all $k \in {0,...,h-1}$: $L_{h_k} \gets I(w_0)$ \Comment{Initialize $L_h$} 
		\State $T_{idle} \gets 0$ \Comment{Initialize the idle time}
		\While {$T_{idle} \le T_{maxIdle}$}
		\State $N \gets \mathit{Neighbors}(w)$ \Comment{Identify the neighbors of $w$}
		\For  {$w^{'} \in N$}
		\State Compute $I(w^{'})$ \Comment{Evaluate the goodness of $w^{'}$} 
		\EndFor 
		\State $w_{best} \gets$ \textit{BestNeighbor($N$)}\Comment{Select the best neighbor in $N$} 
		\State $w_h \gets$ \textit{random.get(\textbf{$L_h$})} \Comment{Randomly select an element in \textbf{$L_h$}}
		\If{$I_{w_{best}}$ $>$ $I_{w_h}$ \textbf{or} $I_{w_{best}}$ $>$ $I_w$}
		\State $w \gets w_{best}$ \Comment{Accept the candidate}
		\State $T_{idle} \gets 0$ \Comment{Reset the idle time}
		\Else 
		\State $w \gets w$ \Comment{Reject the candidate}
		\State $T_{idle} \gets T_{idle}+1$ \Comment{Increase the idle time}
		\EndIf
		\If{$I_w$ $>$ $I_{w_h}$}
		\State $w_{h} \gets w \wedge I_{w_{h}}:=I_w$ \Comment{Update the history list}
		\EndIf
		\EndWhile
		\State \textbf{if} $I_w$ $\ge \sigma$ \textbf{then} \text{$S$.insert(w)} \textbf{end if}  
		\State \textit{$S \gets$ SYCOS$_{\text{BU}}$}$(X^{'}_T,Y^{'}_T)$ \Comment{Recursively call SYCOS$_{\text{BU}}$}
		\EndWhile
		\State \textbf{return} $S$
		
	\end{algorithmic}
	\label{alg:bottomupsearch}
\end{algorithm}

\subsection{iSYCOS: Integrating SYCOS$_{\text{TD}}$ and SYCOS$_{\text{BU}}$} \label{sec:selectionanalysis}
In this section, we integrate SYCOS$_{\text{TD}}$ and SYCOS$_{\text{BU}}$ into a single integrated framework iSYCOS that automatically analyzes and selects the most efficient search approach among SYCOS$_{\text{BU}}$ and SYCOS$_{\text{TD}}$ for a given pair of input time series. 
The selection process consists of two steps: (1) random sampling and (2) selection analysis. 

\subsubsection{Random sampling} On the input time series, we perform random sampling to select multiple data partitions that will be analyzed in the selection analysis step. First, we divide the input time series of length $N$ into $M$ smaller partitions, each of length $\lfloor \frac{N}{M} \rfloor$. Next, we randomly select $m$ partitions from $M$ and use them as the input for the next step. 

\subsubsection{Selecting between SYCOS$_{\text{TD}}$ and SYCOS$_{\text{BU}}$} The integrated algorithm iSYCOS evaluates the efficiency of SYCOS$_{\text{TD}}$ and SYCOS$_{\text{BU}}$ on the $m$ partitions, taking into consideration two factors: the runtime and the size of extracted windows. More specifically, on each data partition $p_i$ obtained from the random sampling, iSYCOS executes SYCOS$_{\text{TD}}$ and SYCOS$_{\text{BU}}$ on $p_i$, and records their runtimes $t^{i}_{\text{TD}}$, $t^{i}_{\text{BU}}$, and the list of correlated windows $W^{i}_{\text{TD}}$, $W^{i}_{\text{BU}}$. Finally, iSYCOS computes the average runtimes $\bar{t}_{\text{TD}}$ and $\bar{t}_{\text{BU}}$ on $m$ partitions.


From the $m$ window lists $W^{i}_{\text{TD}}$ extracted by SYCOS$_{\text{TD}}$, iSYCOS computes the size $\mid$$w_i$$\mid$ for each window $w_i \in W^{i}_{\text{TD}}$. Next, based on the window size, iSYCOS classifies $w_i$ into small or large windows: $w_i$ is small if $\mid$$w_i$$\mid \le \rho \times s_{\max}$ where $\rho \in (0,1]$ is a user-defined parameter, and large otherwise. Finally, iSYCOS counts the number of small windows $n^{\text{S}}_{\text{TD}}$, and the number of large windows $n^{\text{L}}_{\text{TD}}$ for SYCOS$_{\text{TD}}$.  Similarly, $n^{\text{S}}_{\text{BU}}$ and $n^{\text{L}}_{\text{BU}}$ are counted for SYCOS$_{\text{BU}}$. 


Based on the average runtimes and the number of small and large windows, 
iSYCOS evaluates the efficiency of SYCOS$_{\text{TD}}$ and SYCOS$_{\text{BU}}$ on the input time series by computing the efficiency scores for SYCOS$_{\text{TD}}$ and SYCOS$_{\text{BU}}$, weighting two factors, the runtime and the number of extracted windows as 
\begin{equation}
	\text{score}_{\text{TD}} = \alpha \times \frac{1}{\bar{t}_{\text{TD}}}+ (1-\alpha) \times n^{\text{L}}_{\text{TD}}
	\label{eq:effscore_TD}
\end{equation}

\begin{equation}
\text{score}_{\text{BU}} = \alpha \times \frac{1}{\bar{t}_{\text{BU}}}+ (1-\alpha) \times n^{\text{S}}_{\text{BU}}
\label{eq:effscore_BU}
\end{equation}
where $\alpha$ and $(1-\alpha)$ are the weights placed on the runtime, and the number of large windows of SYCOS$_{\text{TD}}$ or the number of small windows of SYCOS$_{\text{BU}}$, respectively. 

The top-down score computed in Eq. \eqref{eq:effscore_TD} is the weighted sum of the inverse of the average runtime, and the number of large windows. The higher the $\text{score}_{\text{TD}} $, the lower the average runtime and the more large windows extracted by SYCOS$_{\text{TD}}$. Similarly for the bottom-up score, the higher the $\text{score}_{\text{BU}}$, the lower the average runtime and the more small windows extracted by SYCOS$_{\text{BU}}$. By adjusting the value of $\alpha$, users can put more weight on the factor that they prefer to be more important, i.e., the runtime or the extracted windows.

However, since the runtime and the number of extracted windows are measured by different scales, we normalize them between $[0-1]$ so that they are weighted properly: 
\begin{equation}
\widetilde{t}_{\text{TD}}=\frac{\bar{t}_{\text{TD}}}{\bar{t}_{\text{BU}}+\bar{t}_{\text{TD}}};\hspace{0.1in}
\widetilde{t}_{\text{BU}}=\frac{\bar{t}_{\text{BU}}}{\bar{t}_{\text{BU}}+\bar{t}_{\text{TD}}}
\label{eq:normalizeruntime}
\end{equation}

\begin{equation}
\widetilde{n}^{\text{L}}_{\text{TD}}=\frac{n^{\text{L}}_{\text{TD}}}{n^{\text{S}}_{\text{BU}}+n^{\text{L}}_{\text{TD}}};\hspace{0.1in}
\widetilde{n}^{\text{S}}_{\text{BU}}=\frac{n^{\text{S}}_{\text{BU}}}{n^{\text{S}}_{\text{BU}}+n^{\text{L}}_{\text{TD}}}
\label{eq:normalizedwindows}
\end{equation}
where $\widetilde{t}_{\text{TD}}$ and $\widetilde{t}_{\text{BU}}$ are the normalized average runtimes, and $\widetilde{n}^{\text{L}}_{\text{TD}}$ and $\widetilde{n}^{\text{S}}_{\text{BU}}$ are the normalized number of large and small windows of SYCOS$_{\text{TD}}$ and SYCOS$_{\text{BU}}$, respectively.

Finally, we replace the normalized factors computed in Eqs. \eqref{eq:normalizeruntime} and \eqref{eq:normalizedwindows} into Eqs. \eqref{eq:effscore_TD} and \eqref{eq:effscore_BU}, and obtain the normalized efficiency scores for SYCOS$_{\text{TD}}$ and SYCOS$_{\text{BU}}$ as
\begin{equation}
	\text{nscore}_{\text{TD}} = \alpha \times \frac{1}{\widetilde{t}_{\text{TD}}}+ (1-\alpha) \times \widetilde{n}^{\text{L}}_{\text{TD}}
	\label{eq:normalizedeffscore_TD}
\end{equation}
\begin{equation}
\text{nscore}_{\text{BU}} = \alpha \times \frac{1}{\widetilde{t}_{\text{BU}}}+ (1-\alpha) \times \widetilde{n}^{\text{S}}_{\text{BU}}
\label{eq:normalizedeffscore_BU}
\end{equation}

Using $\text{nscore}_{\text{TD}}$ and $\text{nscore}_{\text{BU}}$, iSYCOS selects SYCOS$_{\text{TD}}$ over SYCOS$_{\text{BU}}$ if $\text{nscore}_{\text{TD}} > \text{nscore}_{\text{BU}}$, and selects SYCOS$_{\text{BU}}$ over SYCOS$_{\text{TD}}$ otherwise.

Algorithm \ref{alg:iSYCOS} provides the outlines of integrated iSYCOS framework. First, it starts with random sampling to select $m$ partitions for selection analysis (line 1). Lines 2-5 execute SYCOS$_{\text{TD}}$ and SYCOS$_{\text{BU}}$ on each selected data partition and record their runtimes and extracted windows. Lines 6-10 compute the efficiency scores which decide which search method is selected (lines 11-15). 

\begin{algorithm}[!t]
	\footnotesize
	\caption{iSYCOS: Integrated SYCOS}
	\begin{algorithmic}[1]
		\Statex{{\bf Input:}} $(X_T,Y_T)$: \text{pair of time series} 
		\Statex{{\bf Params:} $\sigma$: \text{correlation threshold,}}
		\Statex{\hspace{1.17cm}$s_{\min}$, $s_{\max}$: \text{minimum and maximum window sizes,}} 
		\Statex{\hspace{1.17cm}$m$, $M$: \text{number of selected and total sub-partitions,}} 
		\Statex{\hspace{1.17cm}$\delta_{\text{TD}}$: \text{sliding step of SYCOS$_{\text{TD}}$,}} 
		\Statex{\hspace{1.17cm}$\delta_{\text{BU}}$: \text{moving step of SYCOS$_{\text{BU}}$}}  
		\Statex{\hspace{1.17cm}$\alpha$: \text{weight placed on runtime factor}} 
		\Statex{{\bf Output:} $\mathit{S}$: \text{set of non-overlapping windows with MI $ \ge \sigma$}}
		
		\State $mList \gets \text{randomSampling(M)}$ \Comment{select $m$ partitions from $M$}
		\For  {$m_i \in mList$}
		\State $t_{\text{TD}}^{i}$, $W_{\text{TD}}^{i}$ $\gets$ SYCOS$_{\text{TD}}(m_i)$ \Comment{Top-down search}
		\State $t_{\text{BU}}^{i}$, $W_{\text{BU}}^{i}$ $\gets$ SYCOS$_{\text{BU}}(m_i)$ \Comment{Bottom-up search}
		\EndFor 
		\State $\bar{t}_{\text{TD}} \gets \text{averageRuntime}(t_{\text{TD}}^{i})$ \Comment{top-down average runtime}
		\State $\bar{t}_{\text{BU}} \gets \text{averageRuntime}(t_{\text{BU}}^{i})$ \Comment{bottom-up average runtime}
		\State $n^{\text{S}}_{\text{TD}}$, $n^{\text{L}}_{\text{TD}}$ $\gets$ $\text{countWindows}(W_{\text{TD}}^{i})$, $i=1,...,m$
		\State $n^{\text{S}}_{\text{BU}}$, $n^{\text{L}}_{\text{BU}}$ $\gets$ $\text{countWindows}(W_{\text{BU}}^{i})$, $i=1,...,m$
		\State $\text{nscore}_{\text{TD}}$, $\text{nscore}_{\text{BU}}$ $\gets$ $\text{normalizedScore}(\bar{t}_{\text{TD}}$, $n^{\text{L}}_{\text{TD}}$, $\bar{t}_{\text{BU}}$, $n^{\text{S}}_{\text{BU}},\alpha)$
		\If{$\text{nscore}_{\text{TD}}$ $>$ $\text{nscore}_{\text{BU}}$}
		\State S $\gets$ SYCOS$_{\text{TD}}(X_T,Y_T,\sigma,s_{\min},s_{\max},\delta_{\text{TD}})$ \Comment{select top-down}
		\Else 
		\State S $\gets$ SYCOS$_{\text{BU}}(X_T,Y_T,\sigma,s_{\min},s_{\max},\delta_{\text{BU}})$ \Comment{select bottom-up}
		\EndIf
		\State return S
	\end{algorithmic}
	\label{alg:iSYCOS}
\end{algorithm}

\subsubsection{Space complexity analysis} The time and space complexities of SYCOS$_{\text{TD}}$ and SYCOS$_{\text{BU}}$ depend on the number of feasible windows $n_w$ in the search space, and the time $t_w$ to compute the MI value of a window. Here, we analyze the space complexity of SYCOS$_{\text{TD}}$ and SYCOS$_{\text{BU}}$, and save the discussion of time complexity for Section \ref{sec:implementation}.

For SYCOS$_{\text{TD}}$, the total number of feasible windows in the search space is $n_w=\sum_{i=s_{\min}}^{s_{\max}}(N-i)$ where $N$ is the length of input time series and $s_{\max} \le N$. Thus, the space complexity of SYCOS$_{\text{TD}}$ is $O(N^2)$.

For SYCOS$_{\text{BU}}$, the total number of feasible windows in the search space is $n_w=(s_{\max}-s_{\min})N\sum_{i=1}^{T_{\text{maxIdle}}} 8 \times i$ where $N$ is the length of the input time series, and $s_{\max} \le N$. Since we have $\sum_{i=1}^{T_{\text{maxIdle}}} i \simeq T_{\text{maxIdle}}^2$, the space complexity of SYCOS$_{\text{BU}}$ is $O(N^2 T_{\text{maxIdle}}^2)$.

However, we note that the above complexity analyses is only for the worst case scenarios. In practice, due to the filtering of correlated windows during the search, the numbers of windows explored by SYCOS$_{\text{TD}}$ and SYCOS$_{\text{BU}}$ are much smaller.

\input{corrthreshold}

%% file: corrthreshold.tex
\subsection{Setting the Correlation Threshold}\label{sec:thres}
Since MI measures the dependencies between variables, its magnitude indicates the strength of correlations, i.e., the larger the MI value, the stronger the correlation. However, MI is an unbounded measure, i.e., $0 \le I_w < \infty$. Thus it is difficult to interpret the correlation strength and to properly set the correlation threshold. 
Thus, we now propose a robust method to set the threshold $\sigma$ based on the \textit{normalized MI}:\vspace{0.05in}
\begin{equation}
0 \le  \tilde{I}_{w} = \frac{I_w}{H_w} \le 1
\label{eq:normalizedMI}
\end{equation}
where $I_w$ is the MI, and $H_w$ is the entropy of $w$, defined as
\begin{equation}
H_w=H(X_w;Y_w)\approx\sum_{y_i\in Y_w}\sum_{x_i\in X_w}\hat{p}(x_i,y_i)\log\hat{p}(x_i,y_i)
\end{equation}

In Eq. \eqref{eq:normalizedMI}, the window entropy $H_w$ represents the amount of uncertainty intrinsically present in the time series of $w$. Thus, $\tilde{I}_{w}$ represents the fraction of the window's uncertainty reduced by the shared information $I_w$. The larger the $\tilde{I}_{w}$, the more information is shared between the window's variables, and thus the stronger correlation. The normalized MI $\tilde{I}_{w}$ is always scaled between $[0,1]$, and thus provides an easier way for users to set the threshold $\sigma$.


%% file: noisetheory.tex
\section{Noise-based Pruning in SYCOS}\label{sec:noisetheory}
In this section, we propose a novel MI-based theory to identify noise in the input time series, and show how to incorporate this theory into SYCOS$_{\text{TD}}$ and SYCOS$_{\text{BU}}$ to filter noises and improve the search performance. 
\subsection{MI-based Theory to Identify Noises}
When SYCOS$_{\text{TD}}$ and SYCOS$_{\text{BU}}$ search for correlated windows, they might visit the same data partition multiple times. This is because data partitions can be shared between neighboring windows in SYCOS$_{\text{BU}}$, or between consecutive windows in SYCOS$_{\text{TD}}$. For example, consider the SYCOS$_{\text{BU}}$ search space in Fig. \ref{fig:searchspacebf}. Let $w=[t_s,t_e]$ (blue point) be the current window, and $N_1$ and $N_2$ be its two nearest neighborhoods. The neighboring windows of $w$ in $N_1$ and $N_2$ share data among each other. More specifically, $w^{t}_{N_{1}}=[t_{s_{N_1}}, t_{e_{N_1}}] \in N_1$ is created from $w=[t_s,t_e]$ by extending the end index $t_e$ of $w$ by a $\delta_1$ step. Thus, $w^{t}_{N_1}$ shares with $w$ the data partition $[t_s, t_e]$ which is exactly the window $w$ itself. Similarly, $w^{t}_{N_2}=[t_{s_{N_2}}, t_{e_{N_2}}] \in N_2$ is created from $w$ by enlarging the end index $t_e$ of $w$ by a $\delta_2$ step, with $\delta_2 > \delta_1$. Thus, $w^{t}_{N_2}$ shares with $w$ a partition $[t_s,t_e]$, while sharing with $w^{t}_{N_1}$ a partition $[t_s, t_{s_{N_2}}]$. 
Similarly, consecutive windows in SYCOS$_{\text{TD}}$ also share data among each other. In Fig. \ref{fig:topdownsearchtree}, when SYCOS$_{\text{TD}}$ moves from $w_0=[t_{s_0},t_{e_0}]$ to $w_1=[t_{s_1},t_{e_1}]$ in layer L$_1$, SYCOS$_{\text{TD}}$ creates an overlapping data partition $[t_{s_1}, t_{e_0}]$ that is shared between the two windows. 
  
The shared data partitions create redundant computation and degrade the search performance. To reduce this redundancy, it is important to identify shared data partitions which do not provide \textit{important information} about correlations. 
Such data partitions are considered as \textit{noises}, and can thus be removed from the search space without affecting the final search results. 
The presence of such data partitions is demonstrated in the following example. 

Consider Fig. \ref{fig:mi3} that plots the MI values of a time series pair with different start indices: the blue line starts at index $0$, the red line starts at index $5$, i.e., the data from $0$ to $5$ are not considered in the red line. From Fig. \ref{fig:mi3}, it is clear that by excluding the data partition $[0-5]$ from the search, the MI values of subsequent windows increase and are generally larger than when including it. Thus, the data partition $[0-5]$ provides no information about correlations between the times series, and can be considered as \textit{noise} that should be eliminated from the search exploration. We exploit the MI properties to establish a \textit{noise identification principle} below.

\textit{Definition 6.1} (\textit{Mixture distribution}) Let $X$ and $U$ be discrete random variables with the corresponding p.m.fs $p_X(x)$ and $p_U(u)$. Let $Z$ be a new random variable which is drawn from the same distribution as $X$ with probability $\theta$ and from the same distribution as $U$ with probability $1-\theta$ for a given $\theta \in [0,1]$. Then $Z$ is said to have a mixture distribution between $p_X(x)$ and $p_U(u)$ and is written as $Z = X \odot_\theta U$.
\begin{theorem}\label{theorem:noise}
	Let $X$, $Y$, $U$, $V$ be discrete random variables and $p_X(x)$, $p_Y(y)$, $p_U(u)$, and $p_V(v)$ be their corresponding p.m.fs.  Let $Z=X \odot _\theta U$ and $W=Y \odot_\eta V$ where $\odot$ denotes the mixture of two variables. Assume that, except for $X$ and $Y$, all other variables are mutually independent, i.e., ($U \perp V) \wedge (X \perp U) \wedge (X \perp V) \wedge (Y \perp U) \wedge (Y \perp V)$. 
	Then $I(X;Y) \ge I(Z;W)$. 
	\label{lem:noise1}
\end{theorem}
\begin{proof}
	$Z$ and $W$ are the two mixed variables: $Z=X \odot_\theta U$ and $W=Y \odot_\eta V$. 
	Then, for a value of $x$ drawn according to $p_X(x)$ and a value of $u$ drawn according to $p_U(u)$, we can write the probabilities for $Z$ as follows:
	\begin{align}\vspace{-0.5in}
	\small
	p_Z(x)&=P(Z =X)p_X(x) = \theta p_X(x) \\
	p_Z(u)&=P(Z =U)p_U(u)=(1-\theta) p_U(u)
	\end{align}	
Similarly, we have:
	\begin{align}
	\small
	p_W(y)&=P(W =Y)p_Y(y) = \eta p_Y(y)\\
	p_W(v)&=P(W =V)p_V(v)=(1-\eta) p_V(v)
	\end{align}
	Then, we have the following joint probabilities:
	\begin{align}
	\small 
	p_{Z,W}(x,y)&=\theta \eta p_{X,Y}(x,y)\\
	p_{Z,W}(x,v)&=\theta (1-\eta) p_{X,V}(x,v)\\
	p_{Z,W}(u,y)&=(1-\theta) \eta p_{U,Y}(u,y)\\
	p_{Z,W}(u,v)&=(1-\theta) (1- \eta) p_{U,V}(u,v)
	\end{align} 
	We have the MI between $X$ and $Y$ as
	\begin{equation}
	\small 
	I(X;Y)=\sum_{y}\sum_{x}p_{X,Y}(x,y)\log\frac{p_{X,Y}(x,y)}{p_X(x)p_Y(y)}
	\end{equation} 
	And the MI between $Z$ and $W$ as 
	\begin{equation}
	\small 
	I(Z;W)=\sum_{w} \sum_{z} p_{Z,W}(z,w)\log\frac{p_{Z,W}(z,w)}{p_Z(z)p_W(w)}
	\label{eq:MIconcatenate}
	\end{equation}
	\\
	Since $Z$ can take the values in $\cal{R}_X$ if $z$ is drawn from $X$, and in $\cal{R}_U$ if $z$ is drawn from $U$ (similarly for $W$), then from Eq. \eqref{eq:MIconcatenate}, it follows that:
	\begin{equation}
	\small
	\begin{split}
I(Z;W)&=\sum_{w \in \cal{R}_Y} \sum_{z \in \cal{R}_X} p_{Z,W}(x,y)\log\frac{p_{Z,W}(x,y)}{p_Z(x)p_W(y)} \\
	& + \sum_{w \in \cal{R}_Y} \sum_{z \in \cal{R}_U} p_{Z,W}(u,y)\log\frac{p_{Z,W}(u,y)}{p_Z(u)p_W(y)} \\
	& + \sum_{w \in \cal{R}_V} \sum_{z \in \cal{R}_X} p_{Z,W}(x,v)\log\frac{p_{Z,W}(x,v)}{p_Z(x)p_W(v)} \\
	& + \sum_{w \in \cal{R}_V} \sum_{z \in \cal{R}_U} p_{Z,W}(u,v)\log\frac{p_{Z,W}(u,v)}{p_Z(u)p_W(v)} \\
	&= \sum_{y \in \cal{R}_Y} \sum_{x \in \cal{R}_X} \theta \eta p_{X,Y}(x,y)\log\frac{\theta \eta  p_{X,Y}(x,y)}{\theta p_X(x) \eta p_Y(y)} \\
	& \hspace*{-1.2cm}+ \sum_{y \in \cal{R}_Y} \sum_{u \in \cal{R}_U} (1-\theta)\eta p_{U,Y}(u,y)\log\frac{(1-\theta)\eta p_{U,Y}(u,y)}{(1-\theta)p_U(u) \eta p_Y(y)} \\
	&\hspace*{-1.2cm} + \sum_{v \in \cal{R}_V} \sum_{x \in \cal{R}_X} \theta (1-\eta)p_{X,V}(x,v)\log\frac{\theta (1-\eta) p_{X,V}(x,v)}{\theta p_X(x) (1-\eta)p_V(v)} \\
	&\hspace*{-1.2cm} + \sum_{v \in \cal{R}_V} \sum_{u \in \cal{R}_U} (1-\theta)(1-\eta)p_{U,V}(u,v)\log\frac{(1-\theta)(1-\eta)p_{U,V}(u,v)}{(1-\theta)p_U(u)(1-\eta)p_V(v)} \\
	\end{split}	
	\label{eq:MIconcatenate1}
	\end{equation}
	Eq. \eqref{eq:MIconcatenate1} can be rewritten as
	\begin{equation}
	\small 
	\begin{split}
	I(Z;W)&=\theta \eta I(X;Y) + (1-\theta)\eta I(U;Y)\\
	&+\theta(1-\eta)I(X;V)+(1-\theta)(1-\eta)I(U;V)
	\end{split}	
	\label{eq:MIconcatenate2}
	\end{equation}
	Since we assume: ($U \perp V) \wedge (X \perp U) \wedge (X \perp V) \wedge (Y \perp U) \wedge (Y \perp V)$, this leads to: $I(U;Y)=0 \wedge I(X;V)=0 \wedge I(U;V)=0$. 
	\\
	Thus, Eq. \eqref{eq:MIconcatenate2} becomes
	\begin{equation}
	\small
	I(Z;W)=\theta \eta I(X;Y) 
	\label{eq:MIconcatenate3}
	\end{equation}
	where $\theta \le 1  \wedge \eta \le 1$. This leads to: $ I(X;Y) \ge I(Z;W) $.
\end{proof}
Theorem \ref{lem:noise1} says that, if $U$ and $V$ are independent from each other and from $X$ and $Y$, then adding them to $X$ and $Y$ will bring more uncertainty to $(X,Y)$, in other words, they reduce the shared information $I(X;Y)$. Based on Theorem \ref{lem:noise1}, we define \textit{noise} as follows.
\\
\textit{Definition 6.2} (\textit{Noise}) Let $(X_T,Y_T)$ be a pair of time series, $\sigma$ be the correlation threshold, and $\tau$ ($0 \le \tau < \sigma$) be the \textit{noise threshold}. Consider two consecutive windows $w_{X,Y}$ and $w^{'}_{X,Y}$ of $(X_T,Y_T)$, and $w^{''}_{X,Y}=w_{X,Y} \odot w^{'}_{X,Y}$ is the mixture of $w_{X,Y}$ and $w^{'}_{X,Y}$. Assume that $I_{w_{X,Y}} > 0$. Then $w^{'}_{X,Y}$ is called \textit{noise} w.r.t.  $w_{X,Y}$ iff $I_{w^{'}_{X,Y}} <  \tau$ $\wedge$ $I_{w^{''}_{X,Y}} < I_{w_{X,Y}}$.

The following sections apply the \textit{noise} identification principle in Theorem \ref{theorem:noise} and Def 6.2 to prune the search space of SYCOS$_{\text{TD}}$ and SYCOS$_{\text{BU}}$, to speed up the search.

\input{applyNoiseTheoryToTD}

\input{applyNoiseTheoryToBU}

%% file: applyNoiseTheoryToTD.tex
\subsection{Noise-Based Pruning for SYCOS$_{\text{TD}}$}\label{sec:improveTD} 
Recall that SYCOS$_{\text{TD}}$ uses a sliding window technique to identify correlations in each layer. For example, consider the windows $w_0$, $w_1$, $w_2$ at L$_1$ in Fig. \ref{fig:topdownsearchtree}. Here, $w_0$ is an uncorrelated window, thus it is shifted by a $\delta$-step to form $w_1$, creating an overlapping data partition $[t_{s_1},t_{e_0}]$ between $w_0$ and $w_1$. The next window $w_1$ is also uncorrelated, and is shifted to form $w_2$, creating the overlapping part $[t_{s_2},t_{e_1}]$. 
By shifting one window to create another window, SYCOS$_{\text{TD}}$ creates overlapping data that will be revisited multiple times during the search process. 
To avoid multiple shiftings where overlapping data are repeatedly visited, the noise theory is applied to shorten this shifting step. 

More specifically, when $w_1$ is created by shifting $w_0$ by a $\delta$-step, $w_1$ is divided into two sub-windows: $w_1=w_o \odot w_{\delta}$ where $w_o=[t_{s_1}, t_{e_0}]$ represents the overlapping data, and $w_{\delta}=[t_{e_0}, t_{e_1}]$ represents the newly added data, and $\odot$ is the concatenation operator. We apply the noise theory to $w_1$ and evaluate its goodness in $2$ steps. First, $I_{w_1}$, $I_{w_o}$ and $I_{w_{\delta}}$ are computed separately. Next, $I_{w_1}$ and  $I_{w_o}$ are compared against the correlation threshold $\sigma$, while $I_{w_{\delta}}$ is compared against the noise threshold $\tau$. 
Assume that by applying Theorem \ref{lem:noise1}, we conclude that $w_{\delta}$ is \textit{noise} w.r.t. $w_o$, i.e., $I_{w_{\delta}} < \tau$ $\wedge$ $I_{w_1} < I_{w_o}$, then \textit{noise} is detected in $w_1$. In this case, instead of shifting $w_1$ to create $w_2$, the search can skip $w_1$ entirely, and move directly to $w_3$ whose start index $t_{s_3}$ is right after the end index $t_{e_1}$ of $w_1$, as in Fig. \ref{fig:topdownsearchtree}. In contrast, if \textit{noise} is not detected in $w_1$, SYCOS$_{\text{TD}}$ shifts $w_1$ to $w_2$ as usual.
By using the noise detection mechanism, the entire shifting process can be avoided, therefore speeding up the search performance. 

However, to ensure the noise theory does not remove important data by chance, we require that \textit{noise} should be consecutively detected $p$ times before canceling the shifting process. For instance, in the example above, after identifying that \textit{noise} is present in $w_1$, SYCOS$_{\text{TD}}$ cancels the shifting process only if \textit{noise} is also detected in the next $(p-1)$ consecutive time windows.

%% file: applyNoiseTheoryToBU.tex
\subsection{Noise-Based Pruning for SYCOS$_{\text{BU}}$}\label{sec:improveBU} 

Recall that SYCOS$_{\text{BU}}$ starts with an initial solution $w_0$, and tries to improve $w_0$ by searching for better solutions in nearby neighborhoods. To prune irrelevant data partitions and improve the search performance during the neighborhood exploration process, we apply the noise theory using the example in Fig. \ref{fig:searchspacebf} as below. 

Assume that $w$ is the current solution of SYCOS$_{\text{BU}}$, and $w^{l}_{N_1}$, $w^{l}_{N_2}$ are its neighbors when the search moves leftwards. 
In the first exploration, the neighbor $w^{l}_{N_1}$ is considered. Since $w^{l}_{N_1}$ is created by extending the start index of $w$ by a $\delta_1-$step to the left, we have: $w^{l}_{N_1}= w_{\delta_1} \odot w$ where $w_{\delta_1}$ is the extension to be added to $w$. Also assume that by applying the noise theory to $w$, $w_{\delta_1}$, and $w^{l}_{N_1}$, we conclude that $w_{\delta_1}$ is \textit{noise} w.r.t. $w$, i.e., $I_{w_{\delta_1}} < \tau$ $\wedge$ $I_{w^{l}_{N_1}} < I_w$. In this case, \textit{noise} is detected in the leftwards direction of $w$, and thus, it is not promising to explore further its leftwards neighborhoods. Hence, in the next exploration, SYCOS$_{\text{BU}}$ does not consider $w^{l}_{N_2}$, and omits the leftwards direction entirely from its search. 

Similarly to SYCOS$_{\text{TD}}$, to avoid the risk of removing important data, we stop the search exploration towards a direction only after $p$ noise detections. In the example above, after identifying that $w_{\delta_1}$ is noise w.r.t. $w$, SYCOS$_{\text{BU}}$ only stops exploring leftwards after $p$ noise detections.

%% file: distributedSYCOS.tex
\section{Efficient and Distributed iSYCOS}\label{sec:distributedSYCOS}
In the previous section, we discussed how to integrate the noise theory into iSYCOS to prune irrelevant data partitions and reduce the search space of SYCOS$_{\text{TD}}$ and SYCOS$_{\text{BU}}$. In this section, we propose an efficient MI computation technique to reduce the computation redundancy across windows, further improving the search performance. Furthermore, we also propose a distributed version of iSYCOS that leverages the parallelism of a computing cluster using Apache Spark to search for correlations in big datasets. 

\subsection{Efficient and Incremental MI Computation} \label{sec:implementation}

\subsubsection{Box-assisted algorithm} Recall that to compute the MI of a window, the \textit{KSG} estimator needs to estimate the population density of data points in that window. To do that, for each data point $p_i=(x_i,y_i)$ in the considered window, it first searches for the $k$ nearest neighbors of $p_i$, and then counts the number of marginal points $n_x$, $n_y$ that fall into the marginal $k$-nearest distances in each dimension $d_x$, $d_y$ \cite{kraskov2004estimating}. Among all the computations, $k$-nearest neighbor (\textit{knn}) search is the most expensive operator. Therefore, we design an efficient data structure to optimize the \textit{knn} search using the boxed-assisted method \cite{vejmelka2007mutual}. 

In the box-assisted method, the search space is divided into equal size boxes where each data point is projected into exactly one box. Each box maintains a list of points belonging to that box. When searching for the $k$-nearest neighbors of point $p_i$, first the box containing $p_i$ is found. Then, the search starts from that box and extends to nearby boxes to find the $k$ nearest points. Next, the distances to the $k^{\text{th}}$-neighbor in each dimension $d_x$, $d_y$ are determined, and the marginal points $n_x$, $n_y$ are computed by counting the number of points falling within these distances. The values of $n_x$ and $n_y$ will be used in Eq. \ref{formula:KSGmi} to estimate the final MI.

Fig. \ref{fig:influencedregion1} illustrates the boxed-assisted method and how it is used in KSG to compute the MI of a window that contains $7$ data points $p_1,...,p_7$. Consider point $p_1$ (red) in the figure. Assume that $k=2$ is nearest neighbor parameter, and \textit{maximum norm}\footnote{\small $L_{\infty}$: $d(p_i,p_j)=\parallel (d_x,d_y) \parallel_{\max}=\max(\parallel x_i-x_j\parallel,\parallel y_i-y_j\parallel)$} is the distance metric. In this setting, the two nearest neighbors of $p_1$ are $p_2$ and $p_3$ (green), and the nearest distances in each dimension are $dx$ and $dy$ (shown in the figure). The nearest distances allow KSG to form the marginal regions (in gray shade), from which the marginal counts are computed. In this case, for point $p_1$, the marginal counts are $n_x=3$ (including $p_2,p_3,p_5$), and $n_y=3$ (including  $p_2,p_3,p_4$). A similar process is applied to other points in the window, and the final MI is computed by accumulating the obtained $n_x$ and $n_y$ values.

\begin{figure}[!t]
	\centering
	\begin{subfigure}{.24\textwidth}
		\includegraphics[width=0.9\linewidth]{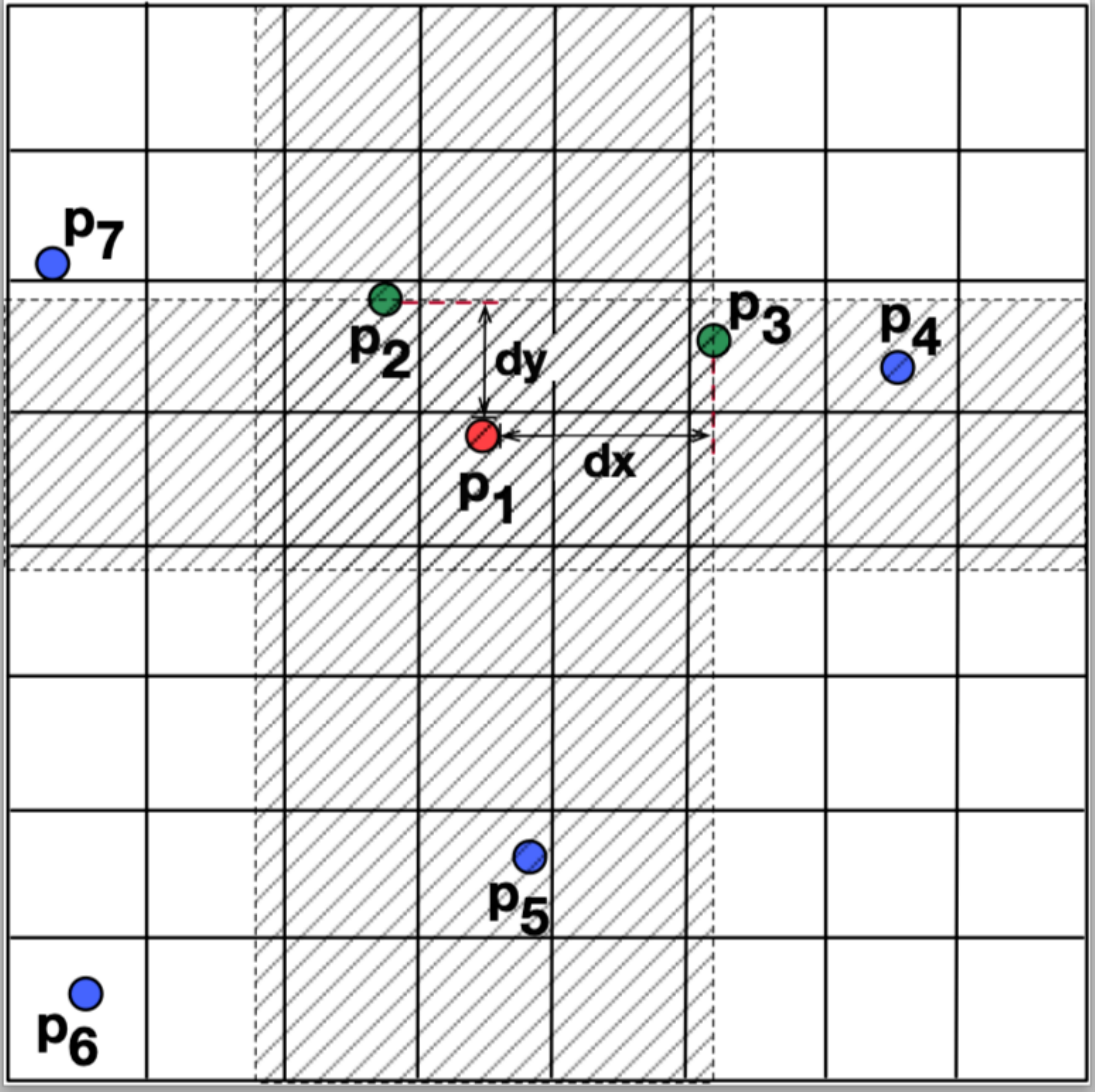}	
		\caption{\footnotesize{Box-assisted method}}
		\label{fig:influencedregion1}
	\end{subfigure}
	\begin{subfigure}{.24\textwidth}
		\includegraphics[width=0.9\linewidth]{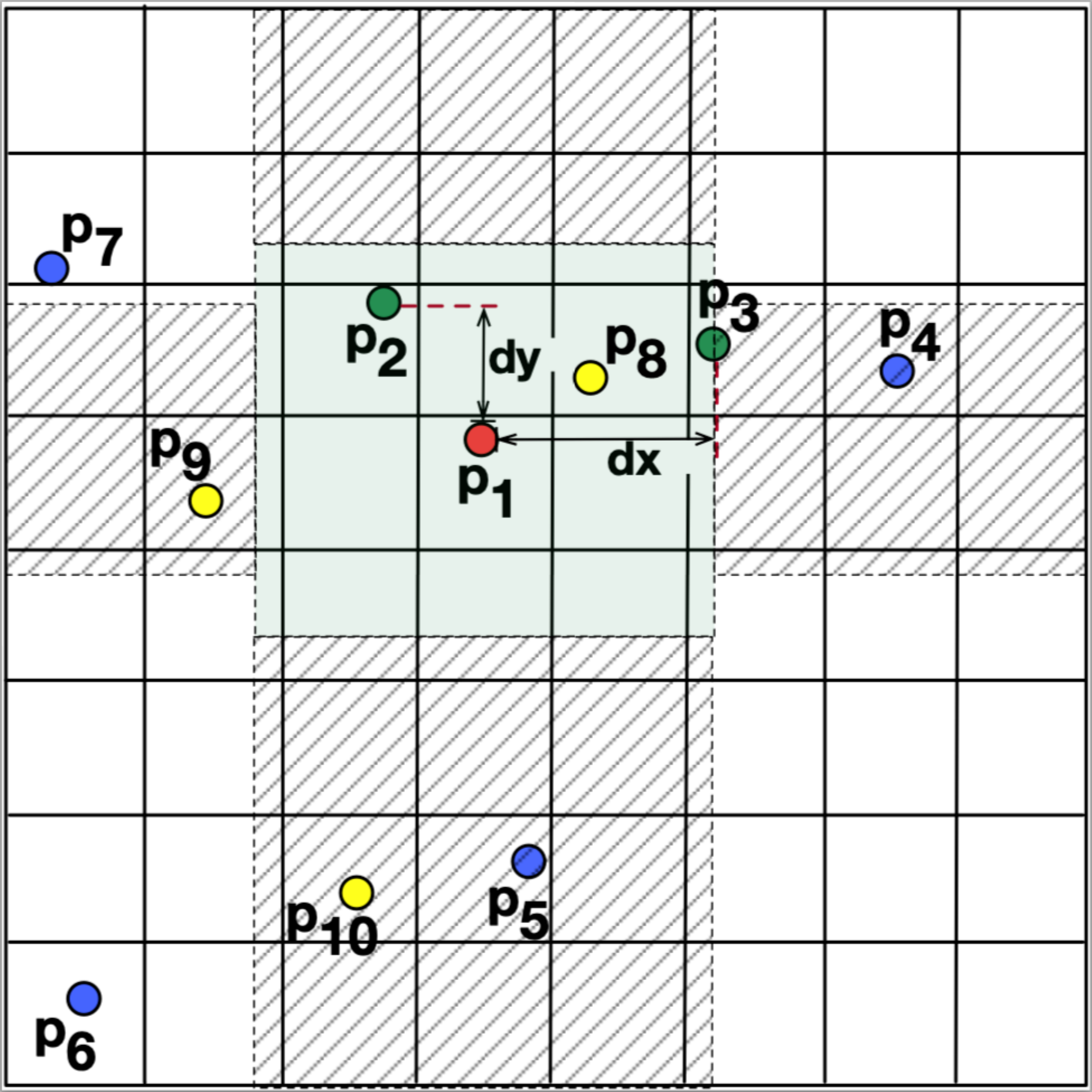}
		\caption{\footnotesize{IR (green) and IMR (gray)}}
		\label{fig:influencedregion2}
	\end{subfigure}
	\caption{Efficient and incremental MI computation} 
	\label{fig:influencedregion}
\end{figure}

\subsubsection{Incremental MI computation with influenced regions} In Section \ref{sec:noisetheory}, we show how the noise theory could avoid revisiting the same data partition multiple times. However, the MI computation across nearby windows can still be repeated. For example, in Fig. \ref{fig:searchspacebf}, when SYCOS$_{\text{BU}}$ moves leftwards, it forms the neighbor window $w^{l}_{N1}=w_\delta  \odot w$ by adding a new data partition $w_\delta$ to the current window $w$. Here, the shared data between $w$ and $w^{l}_{N1}$ are exactly $w$ itself. 
Similarly, in SYCOS$_{\text{TD}}$, the shifting from $w_1$ to $w_2$ in Fig. \ref{fig:topdownsearchtree} creates three different groups of data: \textit{group 1} from $s_1$ to $s_2$ contains the data points of $w_1$  to be removed from $w_2$, \textit{group 2} from $s_2$ to $e_1$ is overlapping data between $w_1$ and $w_2$, and \textit{group 3} from $e_1$ to $e_2$ is newly added data to $w_2$. 

To reuse the computations on overlapping data across windows, we design an optimized boxed-assisted method, ensuring that new computation is only performed on new data, while on overlapping data, only those affected by changes (by removing or adding data) are re-evaluated. 
Note that the introduced changes can be either changing the $k$-nearest neighbors of existing points or changing the marginal counts $n_x$, $n_y$. To keep track of the previous computation, we introduce the concepts of \textit{influenced region} and \textit{influenced marginal region}.

\textit{Definition 7.1} (\textit{Influenced region}) An \textit{influenced region} (\textit{IR}) of point $p_i=(x_i,y_i)$ is a square bounding box $R_{i}=(l_{i},r_{i},b_{i},t_{i})$, where $l_{i},r_{i},b_{i},t_{i}$ are its left-, right-, bottom-, and top-most indices, respectively, which are computed as $l_{i}=x_i-d$, $r_{i}=x_i+d$, $b_{i}=y_i+d$, $t_{i}=y_i-d$ where $d=\max(d_x,d_y)$.

\textit{Definition 7.2} (\textit{Influenced marginal region}) The \textit{influenced marginal region(s)} (\textit{IMR}) of point $p_i$ is (are) the marginal region(s) located within the nearest distance $d_i$ in each dimension. 

Fig. \ref{fig:influencedregion2} illustrates the \textit{influenced region} and \textit{influenced marginal region} concepts. The \textit{influenced region} of $p_1$ is the square colored in green, and the \textit{influenced marginal regions} are those with gray shade in either dimension. 

\begin{lem}\label{lem3}
	Given a window $w_i$ and a data point $p \in w_i$, a new point $o$ inserted into $w_i$ will become the new $k^{th}$-neighbor of $p$ iff $o$ is within IR of $p$. 
	\vspace{-0.08in}
\end{lem}

\begin{lem}\label{lem4}
	Given a window $w_i$ and a data point $p \in w_i$, an existing point $o$ deleted from $w_i$ will change the $k$ nearest points of $p$ iff $o$ is within IR of $p$. 
	\vspace{-0.08in}
\end{lem}

\begin{lem}\label{lem5}
	Given a window $w_i$ and a data point $p \in w_i$, a new point $o$ inserted into $w_i$ will increase the marginal count $n_x$ (or $n_y$) of $p$ iff $o$ is within \textit{IMR}$_x$ (or \textit{IMR}$_y$) of $p$.
	\vspace{-0.08in}
\end{lem}

\begin{lem}\label{lem6}
	Given a window $w_i$ and a data point $p \in w_i$, an existing point $o$ deleted from $w_i$ will reduce the marginal count $n_x$ (or $n_y$) of $p$ iff $o$ is within \textit{IMR}$_x$ (or \textit{IMR}$_y$) of $p$.
	\vspace{-0.08in}
\end{lem}

\begin{proof}
	Proofs of Lemmas \ref{lem3}, \ref{lem4}, \ref{lem5}, \ref{lem6} are straightforward, thus omitted.   
\end{proof}

Lemmas \ref{lem3}, \ref{lem4}, \ref{lem5}, \ref{lem6} display unique properties of \textit{IRs} and \textit{IMRs}. An \textit{IR} maintains an area where any point $p_j$ either being added to or removed from this region will change the $k$ nearest points of $p_i$. In this case, a new $k$-nearest neighbors search is required for $p_i$. Instead, an \textit{IMR} maintains an area where any point $p_j$ either being added to or removed from it will change the marginal counts of $p_i$. In this case, the marginalized neighbors of $p_i$ have to be recounted. 

Fig. \ref{fig:influencedregion2} illustrates how changes are introduced and managed. For simplicity, we only discuss the cases when new points are added into a previous computation. Changes introduced by removing points can be handled in a similar way. Assume that at time $t_1$, a new point $p_8$  is added to the current window and falls into the IR of $p_1$. The addition of $p_8$ changes the $k^{th}$-nearest neighbor of $p_1$. Thus, a new nearest neighbor search for $p_1$ is required. 

At time $t_2$, a new point $p_9$ arrives and falls into the $y$-marginal influenced region of $p_1$, for which it will alter the marginal count $n_y$  (but no new $k$-nearest neighbor search is required in this case). Similarly, a new point $p_{10}$ will increase the marginal count $n_x$. In these cases, only a recount of $n_x$ or $n_y$ is performed. 

As the result of our \textit{optimized box-assisted computation} method, for each window, only a minimum search region (containing new points) and a minimum update region (containing points affected by added and removed points) require additional computation. The rest is reused, thus minimizing the computational cost.

\subsubsection{Time complexity analysis}
The time complexities of SYCOS$_{\text{TD}}$ and SYCOS$_{\text{BU}}$ depend on the number of feasible windows $n_w$ in the search space, and the time $t_w$ to compute the MI value of each window. 

Recall that the k-nearest neighbor (knn) search is the most expensive operator in the MI computation. Using the boxed-assisted method for efficient knn search, the time complexity of the MI computation is $t_w \sim O(n\log n)$ where $n$ is the window size. Moreover, we have the number of feasible windows $n_w \sim O(N^2)$ for SYCOS$_{\text{TD}}$, and $n_w \sim O(N^2 T_{\text{maxIdle}}^2)$ for SYCOS$_{\text{BU}}$.

Thus, the time complexity of SYCOS$_{\text{TD}}$ is $O(N^2 n\log n)$, and of SYCOS$_{\text{BU}}$ is $O(N^2 T_{\text{maxIdle}}^2 n\log n)$, where $n$ is the average window size.

However, we note that these are worst case complexities. In practice, thanks to the filtering of correlated windows during the search, and the incremental MI computation technique, the time complexities of SYCOS$_{\text{TD}}$ and SYCOS$_{\text{BU}}$ are much smaller (refers to the experiments in Section \ref{sec:experimentquantitative}).

\subsection{Distributed iSYCOS using Apache Spark}
The distributed iSYCOS distributes the computation among worker nodes in a Spark cluster to speed up the search. This distributed computation applies to both the selection analysis (to select the most efficient search method among SYCOS$_{\text{BU}}$ or SYCOS$_{\text{TD}}$), and the correlation search. Algorithm \ref{alg:parallel} provides the pseudo-code of the distributed iSYCOS, with the detailed explanation below.

In the selection analysis step (function \textit{algSelection}, lines 4-10), first, iSYCOS samples $m$ data partitions from the input time series (lines 5-7). This data sampling step is performed at the master node since distribution is not needed for a fast operation. Next, iSYCOS \textit{maps} each sampled data partition $p_i$ to a worker node (line 8), and executes SYCOS$_{\text{BU}}$ and SYCOS$_{\text{TD}}$ on $p_i$ (lines 1-3). The performance statistics obtained on $p_i$, including the runtimes and the extracted windows, are collected by the Spark function \textit{collect()} to the master node. Finally, the efficiency scores $\textit{nscore}_{\textit{BU}}$ and $\textit{nscore}_{\textit{TD}}$ are computed at the master node (line 9), and a search method is selected using the policies in Section \ref{sec:selectionanalysis}.

Using the selected search method, iSYCOS proceeds to distribute the correlation computation by first dividing the original time series into multiple \textit{overlapping data chunks} (line 14). 
The overlapping data ensures the correlation analysis is contiguous between data chunks, and is equal to the maximum window size $s_{\max}$. Next, iSYCOS \textit{maps} each overlapping data chunk to a worker node and invokes the selected algorithm, i.e., SYCOS$_{\text{BU}}$ or SYCOS$_{\text{TD}}$, to perform correlation search on that chunk (lines 15-18). Finally, iSYCOS collects the extracted correlated windows from each worker node using the Spark function \textit{collect()}. 

\begin{algorithm}
	\caption{Distributed iSYCOS}
	\label{alg:parallel}
	\footnotesize
	\begin{algorithmic}[1]
		\Statex{{\bf Input:}} $\{X_T,Y_T\}$: \text{pair of time series variables} 
		\Statex \hspace{1.07cm}\text{$s_{\max}$: maximum window size,} 
		\Statex \hspace{1.07cm}\text{$n_p$: number of overlapping data chunks}
		\Statex \hspace{1.07cm}\text{$N$: length of time series}
		\Statex \hspace{1.07cm}\text{lenS: length of the sampled data partition}
		\Statex \hspace{1.07cm}\text{numS: number of the sampled data partitions}
		
		\Statex
		\State def algAnalysis(startIdx, endIdx):
		\State	\hspace{0.5cm}output $\gets$ subprocess.check\_output(startIdx, endIdx)
		\State	\hspace{0.5cm}return output.decode()
		
		\Statex 
		\State def algSelection(N, lenS, numS):
		\State	\hspace{0.5cm}list\_range $\gets$ [[x, x+lenS-1] for x in range(0, N, lenS)]
		\State	\hspace{0.5cm}samples $\gets$ random.sample(list\_range, numS)
		\State	\hspace{0.5cm}samplesRDD $\gets$ sc.parallelize(samples)
		\State	\hspace{0.5cm}t$_{\text{TD}}$, W$_{\text{TD}}$, t$_{\text{BU}}$, W$_{\text{BU}}$ $\gets$ samplesRDD.map(algAnalysis).collect()
		\State	\hspace{0.5cm}nscore$_{\text{TD}}$, nscore$_{\text{BU}}$ $\gets$ calculateScores(t$_{\text{TD}}$, W$_{\text{TD}}$, t$_{\text{BU}}$, W$_{\text{BU}}$)
		\State \hspace{0.5cm}return nscore$_{\text{TD}}$, nscore$_{\text{BU}}$ 
		
		\Statex 
		\State if \_\_name\_\_ == "\_\_main\_\_":
		\State \hspace{0.5cm}sc $\gets$ SparkContext(conf=SparkConf()); 
		\State \hspace{0.5cm}nscore$_{\text{TD}}$, nscore$_{\text{BU}}$ $\gets$ algSelection(N, lenS, numS)
		\State \hspace{0.5cm}dataRDD $\gets$ sc.parallelize([(i-s$_{\max}$,i+N/n$_p$) for  $i$ 
		\Statex \hspace{5cm}  in xrange(0,N,N/n$_p$)]); 
		\State \hspace{0.5cm}if $\text{nscore}_{\text{TD}}$ $>$ $\text{nscore}_{\text{BU}}$ then:
		\State \hspace{1cm}corrWindows $\gets$ dataRDD.\textbf{map}(SYCOS$_{\text{TD}}$).collect(); 
		\State \hspace{0.5cm}else:
		\State \hspace{1cm}corrWindows $\gets$ dataRDD.\textbf{map}(SYCOS$_{\text{BU}}$).collect();
	\end{algorithmic}
\end{algorithm}

%% file: experiment.tex
\section{Experimental Evaluation}\label{sec:experiment}
In this section, we evaluate the performance of iSYCOS  in terms of the quality of extracted windows, the runtime efficiency and the accuracy. We use both synthetic and real-world datasets from the energy and the smart city domains. 

\subsection{Baseline methods}
\textit{\textbf{Qualitative evaluation:}} To evaluate the quality of extracted windows, iSYCOS is compared against two baselines: (1) the traditional correlation metric Pearson Correlation Coefficient (PCC) \cite{pearson1895notes}, and (2) Fast Subsequence Search (MASS) \cite{rakthanmanon2012searching}, a well-known method for time series sub-sequences matching. We use MASS2 \cite{mass2} which provides a faster implementation than the original MASS. 
\\
\textit{\textbf{Quantitative evaluation:}} To be quantitatively comparable to iSYCOS, a baseline method must have similar objectives, i.e., it can detect correlations with multiple temporal scales. In this regard, we exclude seem-to-be-similar methods such as \textit{Data Polygamy} \cite{chirigati2016data} (because it does not consider multi-scale correlations), DTW \cite{salvador2007toward} and MASS \cite{rakthanmanon2012searching} (because they consider the problem of time series similarity search instead of correlation search). In fact, there are no existing methods that perform window-based correlation search as iSYCOS. For this reason, iSYCOS runtime efficiency is evaluated by comparing SYCOS$_{\text{BU}}$ and SYCOS$_{\text{TD}}$ against each other, using their different versions: with and without MI optimization and noise pruning. Furthermore, we also analyze the effectiveness of iSYCOS in selecting the most suitable method between SYCOS$_{\text{BU}}$ and SYCOS$_{\text{TD}}$ on different datasets. 

\subsection{Datasets}
To test the robustness of iSYCOS, we generate synthetic datasets containing different types of relations, including both linear and non-linear, monotonic and non-monotonic, functional and non-functional relations. 
To examine iSYCOS effectiveness in identifying correlations in real-life scenarios, we use real-world data from the energy and the smart city domains. The smart city data are obtained from the NYC Open Data portal \cite{nycopen}, containing more than 1,500 spatio-temporal datasets and providing a wide range of information about New York City. The energy data are obtained from an energy trading company in Belgium, under an NDA\footnote{\small Due to the Non Disclosure Agreement (NDA), we do not disclose the company name and its datasets.} agreement, containing energy production data from multiple wind farms. Besides, we use another open energy dataset from Kaggle \cite{energykaggle}, providing energy consumption of electrical appliances in residential buildings in Belgium.

\subsection{Parameter settings}
iSYCOS depends on the setting of $4$ parameters: correlation threshold $\sigma$, noise threshold $\varepsilon$, minimum window size $s_{\min}$, and maximum window size $s_{\max}$. Among them, $\sigma$, $s_{\min}$, and $s_{\max}$ are user-defined parameters, while $\varepsilon$ is a hyper parameter. The $\sigma$ value determines the strength of extracted correlations (larger $\sigma$, stronger correlations), while the value of $\varepsilon$ determines the degree of noise allowed over correlated data. For example, a fraction $\varepsilon / \sigma=0.25$ means that data with MI less than $\frac{1}{4}$ of the correlation threshold are considered as noise and thus, unpromising to explore further.  On the other hand, the values of $s_{\min}$ and $s_{\max}$ are context dependent, and are set based on domain knowledge. Given an application domain, it is usually intuitive how small or large a window could be. For example, when users analyze weather related data, they might decide that the longest duration of a weather event is \textit{two weeks}, and thus set the size of $s_{\max}$ to \textit{two weeks}.

In our experiments, we use the normalized MI (scaled between $[0,1]$) introduced in Section \ref{sec:thres} to set the value of $\sigma$. The hyper parameter $\varepsilon$ is set equal to $\frac{1}{4} \cdot \sigma$ in all experiments. The ratio $\varepsilon / \sigma = 0.25$ is chosen based on the empirical studies conducted on different datasets, which consistently show that $\varepsilon / \sigma \simeq 0.25$ yields the best trade-off between accuracy and runtime gain. Table \ref{tbl:param} lists the values of $\sigma$, $s_{\min}$, and $s_{\max}$ we use in each dataset. 

\begin{table}[!t]
	\footnotesize
	\caption{Parameters setting}
	\label{tbl:param}
	\vspace{-0.1in}
	\centering
	\begin{tabular}{|c|c|c|}
		\hline 
		Parameter & Energy datasets & Smart city datasets \\ \hline
		$\sigma$ & 0.2 & 0.2  \\  \hline 
		$s_{\min}$ & 30 samples $\simeq$ 30 mins  & 4 samples $\simeq$ 1 hour   \\ \hline
		$s_{\max}$ & 10080 samples $\simeq$ 7 days & 2880 samples $\simeq$ 30 days  \\ \hline
	\end{tabular}
\end{table}

\subsection{Qualitative evaluation}

\input{syntheticdataset}
\input{nycdataset}
\input{energydataset}

\subsection{Quantitative evaluation} \label{sec:experimentquantitative}

\input{runtime}

\input{selectionmethodevaluation}

\input{scalability}

%% file: syntheticdataset.tex
\subsubsection{Using synthetic datasets}

We apply both SYCOS$_{\text{BU}}$ and SYCOS$_{\text{TD}}$ to the synthetic data to test its ability to detect the generated relations. We compare the results to the baselines, i.e., PCC and MASS. Table \ref{tbl:synthetic} shows the types of relations, $y=f(x)$ and $u$ is added noise, recognized by each method.

Our experiments show that both methods SYCOS$_{\text{BU}}$ and SYCOS$_{\text{TD}}$ can identify all types of relations, even when there is noise present in the data. Compared to the baselines, PCC cannot detect non-functional relations, while MASS cannot identify non-linear and non-functional relations. To explain the weak performance of the baselines, we can see that MASS is designed to identify similarities in time series. Therefore, for non-linear and non-functional relations such as a circle, there is no similarity between a variable $x$ that is linearly increasing, and a variable $y$ that follows a round shape, and thus they cannot be identified by MASS.

\begin{table}[!t]
		\footnotesize
		\caption{Identify different types of correlation relations }
		\label{tbl:synthetic}
		\centering	
		\begin{tabular}{|p{3.5cm}|>{\centering} p{1.0cm}|>{\centering} p{1.0cm}|>{\centering} p{0.5cm}|>{\centering} p{0.7cm}|}
			\hline
			Relations: $y=f(x)$ & SYCOS$_\text{BU}$ & SYCOS$_\text{TD}$  & PCC & MASS   \tabularnewline \hline 
			Ind.: $y\sim N(0,1)$\tnote{(*)}, $x\sim N(3,5)$ &\large \checkmark &\large \checkmark &\large \checkmark &\large \checkmark  \tabularnewline \hline
			Linear: $y=2x + u$\tnote{(**)}, $x \in [0,10]$  &\large \checkmark & \large \checkmark &\large \checkmark &\large  \checkmark  \tabularnewline  \hline 
			Exponential: $y=0.01^{x+u}$, $x \in [-10,10]$ &\large \checkmark &\large  \checkmark &\large \checkmark  & \large \checkmark  \tabularnewline \hline 
			Quadratic: $y=x^2+u$, $x \in [-4,4]$ & \large \checkmark &\large  \checkmark & $\Cross$ &\large  \checkmark  \tabularnewline \hline 
			Diamond: $y_1=x+u$, $y_2=8-x+u$, $y_3=-4+x+u$, $y_4=12-x+u$, $x \in [4,8]$ &\large \checkmark &\large  \checkmark &$\Cross$ & \large \checkmark  \tabularnewline \hline 
			Circle: $y=\pm \sqrt{3^2-x^2+u}$, $x \in [-3,3]$ &\large \checkmark &\large  \checkmark & $\Cross$ & $\Cross$  \tabularnewline \hline 
			Sine $y=2*sin(x)+u$, $x \in [0,10]$ &\large \checkmark & \large \checkmark & $\Cross$ & $\Cross$ \tabularnewline \hline 
			Cross: $y_1=x+u$, $y_2=-x+u$, $x \in [-5,5]$ &\large \checkmark  & \large \checkmark & $\Cross$ & $\Cross$  \tabularnewline \hline 
			Quartic: $y=x^4-4x^3+4x^2+x+u$, $x \in [-1,3]$ & \large \checkmark & \large \checkmark & $\Cross$ &\large \checkmark   \tabularnewline \hline 
			Square root:  $y=\sqrt{x}$, $x \in [0,25]$ &\large \checkmark &\large \checkmark &$\Cross$ & \large \checkmark  \tabularnewline 
			\hline
		\end{tabular}
		\begin{tablenotes}
			\item[*] \hspace{0.0cm} *$N(\mu,\sigma)$: normal distribution
			\item[**] \hspace{0.0cm} **$u \sim U(0,1)$: uniform distribution
	\end{tablenotes}
\end{table}

%% file: nycdataset.tex
\subsubsection{Using real-world datasets}
We apply SYCOS$_{\text{BU}}$ and SYCOS$_{\text{TD}}$ to real-world datasets to extract correlated windows. We interpret the insights from extracted correlations, and report a few of them below.

\textbf{The smart city datasets:} We focus on two collections of data related to \textit{transportation} and \textit{weather}. The data are measured in \textit{day}, \textit{hour} and \textit{minute} resolutions.

We analyzed the correlations between the \textit{Weather} and the \textit{Taxi} datasets, and found a strong \textit{negative} correlation (one variable increases while the other decreases) between the \textit{wind speed} and the \textit{number of taxi trips} in the city. 
When associating the correlated windows with their time intervals, we observe that the correlations between these two variables only occurred at times when extreme weather events happened in NYC, such as the occurrence of a hurricane or storm. 
We then ordered the extracted windows by their MI values, and found that many of the top ranked windows are in fact associated with extreme weather conditions, as reported in Table \ref{tbl:topwindows}.
\begin{table}[!t]
	\small
	\renewcommand{\arraystretch}{1.3}
	\caption{Top K windows between Taxi and Wind Speed}
	\vspace{-0.2cm}
	\label{tbl:topwindows}
	\centering
	\begin{tabular}{|p{1.6cm}|p{1.6cm}|p{0.6cm}|p{3.3cm}|}
		\hline
		From & To & MI & Event \\ \hline 
		2012/10/29 &	2012/11/02	& 0,65 & Sandy Hurricane \\ 
		2012/07/27 &	2012/07/28 & 0,6	& Tornado hits NYC \\ 
		2012/01/20	& 2012/01/21 &	0,57 &	Snow storm \\ 
		2013/06/07 & 2013/06/10 & 0.54& Tropical storm Andrea \\ 
		2012/11/10	& 2012/11/11 &	0,54 &	Snow storm \\ 
		2012/12/23 & 2012/12/24 &	0,49 & Storm \\ 
		2012/11/06 & 2012/11/07 &	0,48 & Snow storm \\ 
		2012/09/08 & 2012/09/09 &	0,42 &	Tornado  \\ 
		2011/08/26 & 2011/08/31 & 0,41 & Irene Hurricane \\
		2011/10/29 & 2011/10/31 & 0,40 & Snow storm \\ 
		\hline
	\end{tabular}
	\vspace{-0.2cm}
\end{table}

We test the relation between the \textit{average rain precipitation} and the \textit{number of taxi trips}. In the extracted windows, we found a \textit{negative} correlation between these two variables: a drop in taxi trips associated with abnormally high rain. These correlations also occurred during extreme weather events, like the wind speed. 

We test another pair of variables, the \textit{taxi fare} and the \textit{rain precipitation}. The authors in \cite{chirigati2016data}, when using the same datasets, report a \textit{positive} relationship (both variables are increasing or decreasing together) between the fare taxi drivers earn and the rain precipitation, suggesting that taxi drivers increase earnings when it rains. 
In our findings, we also found this \textit{positive} correlation. Moreover, we found that this correlation is visible only in \textit{hour} resolution, but disappears in \textit{day} resolution. This phenomena might be explained based on the fact that taxi drivers are target earners: taxi drivers have a daily income target, and reach their targets sooner when it rains, after which they quit driving for the remainder of the day, and thus keep the same daily income.

In another finding, we found a \textit{weak positive} correlation between the \textit{taxi fare} and the \textit{trip duration}. The correlation suggests that taxi driver is likely to earn more when the trip duration is longer. Although this correlation seems obvious, it is not trivial why it is only a weak correlation. To explain this weak correlation, we note that trip duration gets longer either because of traffic jam (the fare does not increase) or because the travel distance is long (the fare does increase). If traffic jam is the cause of long trip duration, linking \textit{trip duration} to \textit{taxi fare} is not enough, and thus, explains the \textit{weak} correlation between the two variables. This suggests that other variables, such as \textit{travel distance}, can be of interest to be analyzed together with \textit{taxi fare}. 

%% file: energydataset.tex
\textbf{The energy datasets:}
The energy data are recorded at high frequency (every 2 seconds), thus, we aggregate them into minute resolution before applying iSYCOS.

When analyzing the energy data, we first want to confirm that iSYCOS can find known correlations. Thus, we test the relation between the produced energy and its source, i.e., the \textit{active power} and the \textit{solar irradiation}, and the \textit{active power} and the \textit{wind speed}. The correlated windows extracted from iSYCOS indeed confirm the strong \textit{positive} correlations between the produced energy and its respective source. For example, extracted windows show that the generated power increases when the solar irradiation is high or when the wind speed is strong, and vice versa. This strong correlation holds for the entire time series of the considered variables. 
Moreover, we found that the wind speed is strongly correlated with the rotor speed of the wind turbine. 

When analyzing the relation between the produced power and weather-related variables, we found a \textit{negative} correlation between the power production and the humidity, and a \textit{positive} correlation between the power production and the temperature. This inspires us to test the relationship between the humidity and the temperature. As expected, we found a \textit{negative} correlation between these two variables. These inter-correlations between the three variables are naturally intuitive, since high solar irradiation and low humidity result in high temperature, and vice versa. Similarly, when wind turbines are the energy source, we found a \textit{negative} correlation between the air pressure and the wind speed, and thus, also a \textit{negative} correlation between the air pressure and the produced power. These correlations can be explained based on meteorology, in which the air flows from high pressure area to low pressure area, and thus creates wind. This results in a \textit{negative} correlation between the air pressure and the wind speed. Since the wind speed is directly correlated to the produced energy (through the wind turbine), it explains the negative correlation between the air pressure and the energy production.

A somewhat surprising finding is that we found the positive correlation between the wind direction and the wind speed. As wind direction indicates the direction in which the wind flows (measured in degrees), and wind speed describes how fast the air is moving (measured in km/hour), it is not trivial how to explain the nature of this correlation. However, the discovered correlation might prompt users new opportunities to perform further analysis. As we do not attempt to explain the causality behind the correlation, we refer readers to other articles such as \cite{velazquez2011comparison} for further understanding of this phenomenon. 

Similar meteorological phenomena are also found in the open energy dataset from Kaggle. For example, we confirmed a strongly negative correlation between the room temperature and the humidity, and between the humidity and the air pressure. In contrast, a strongly positive correlation is detected between the temperature and the air pressure.

%% file: runtime.tex
\subsubsection{SYCOS$_{\text{TD}}$ and SYCOS$_{\text{BU}}$ performance comparison}\label{sec:runtime}
We first evaluate the runtime efficiency of iSYCOS by comparing the performance of SYCOS$_{\text{TD}}$ and SYCOS$_{\text{BU}}$ on different datasets. The experiments are run on a VM with 4 vCPUs, 32GB of RAM, and 1024GB of storage.

\begin{table}[!t]
	\centering
	\caption{MI and window sizes of SYCOS$_{\text{TD}}$ and SYCOS$_{\text{BU}}$}
	\label{tbl:MIandWindowTDBU}
	\vspace{-0.2cm}
	\begin{tabular}{|>{\rowmac}p{3.05cm}|>{\rowmac}>{\centering}p{0.77cm}|>{\rowmac}c|>{\rowmac}c<{\clearrow}|} \hline
		Data & Method & Avg. W. Size & Avg. MI \\ \hline
		
		\multirow{2}{*}{Taxi Trips, Traffic Speed} & \bfseries TD & \bfseries  719  & \bfseries  0.84\\
		& BU & 22 & 0.68  \\ \hline
		
		\multirow{2}{*}{Taxi Trips, Rain} & TD & 24 & 0.25 \\
		 & \bfseries  BU & \bfseries  14  & \bfseries  0.35  \\ \hline
		
		\multirow{2}{*}{Taxi Fare, Rain} & TD &  19 & 0.36 \\
		 & \bfseries BU & \bfseries 16 & \bfseries  0.44  \\ \hline
		 
		 \multirow{2}{*}{W. Speed, A. Power} &\bfseries  TD & \bfseries 549  & \bfseries 0.6 \\
		 & BU &  184 & 0.62	  \\ \hline
		 
		 \multirow{2}{*}{R. Speed, A. Power} & \bfseries TD & \bfseries 662  & \bfseries  0.68 \\
		 & BU &  116	& 0.64  \\ \hline
		 
		 \multirow{2}{*}{W. Speed, Rotor Speed} & \bfseries TD & \bfseries 648  & \bfseries  0.57  \\
		 & BU &  131 & 0.52	  \\ \hline
		
		\multirow{2}{*}{W. Direction, Pitch Angle} & TD &  47 & 0.24 \\
		 & \bfseries BU &  \bfseries 16 &\bfseries  0.28  \\ \hline
		
		\multirow{2}{*}{W. Direction, W. Speed} & TD & 86  & 0.35  \\
		 & \bfseries BU & \bfseries 56	& \bfseries 0.37  \\ \hline
	\end{tabular}
\end{table}

Fig. \ref{fig:runtime} shows the runtimes of SYCOS$_{\text{TD}}$ and SYCOS$_{\text{BU}}$, and Table \ref{tbl:MIandWindowTDBU} shows the average MI and window size extracted from each method on different variable pairs of the smart city and the energy datasets. From Fig. \ref{fig:runtime}, we can see that  SYCOS$_{\text{TD}}$ and SYCOS$_{\text{BU}}$ perform differently on different datasets. There are cases where SYCOS$_{\text{TD}}$ outperforms SYCOS$_{\text{BU}}$, and vice versa. For example, SYCOS$_{\text{TD}}$ outperforms SYCOS$_{\text{BU}}$ on the pairs (Taxi Trips, Traffic Speed), (Wind Speed, Active Power), (Rotor Speed, Active Power), and (Wind Speed, Rotor Speed). The average MI of extracted windows from these pairs are significantly higher than other pairs (e.g., 0.52-0.84), indicating that the variables are strongly correlated. The average window sizes of SYCOS$_{\text{TD}}$ are also significantly larger than SYCOS$_{\text{BU}}$. This supports our initial hypothesis that the top-down approach is more efficient when the variables are strongly correlated, and the correlated windows are large. On these pairs of variables, SYCOS$_{\text{TD}}$ obtains an average speedup of 94.27 over SYCOS$_{\text{BU}}$, and the speedup range is [1.4-370].

Conversely, SYCOS$_{\text{BU}}$ outperforms SYCOS$_{\text{TD}}$ on the pairs (Taxi Trips, Rain Precipitation), (Taxi Fare, Rain Precipitation), (Wind Direction, Pitch Angle), (Wind Direction, Wind Speed), and (Pitch Angle, Rotor Speed). It is noticeable that the average MIs of extracted windows from these pairs are significantly lower (e.g, 0.2-0.44) than those where SYCOS$_{\text{TD}}$ is faster, indicating that the variables are weakly or moderately correlated. The average sizes of extracted windows on these pairs are also significantly smaller. This also supports our initial hypothesis that the bottom-up approach is more efficient when the variables are not well correlated, and the correlated windows are small and sparse. On these pairs of variables, SYCOS$_{\text{BU}}$ obtains an average speedup of 5.3 over SYCOS$_{\text{TD}}$, and the speedup range is [1.3-15.5]. We note that however, the average MIs of extracted windows from SYCOS$_{\text{BU}}$ and SYCOS$_{\text{TD}}$ are similar on each variable pair.

In summary, when comparing SYCOS$_{\text{TD}}$ and SYCOS$_{\text{BU}}$, we conclude that SYCOS$_{\text{TD}}$ should be used to search for correlated windows when the variables are strongly correlated, while SYCOS$_{\text{BU}}$ should be used when the variables are weakly or moderately correlated.

\input{figures.tex}

\subsubsection{Noise pruning and MI optimization}		
Next, we evaluate how effective the proposed optimizations, including the MI-based noise pruning (Section \ref{sec:noisetheory}) and the MI optimization computation (Section \ref{sec:implementation}) are in improving the performance of SYCOS$_{\text{TD}}$ and SYCOS$_{\text{BU}}$. Particularly, we compare the runtimes of four different versions of SYCOS$_{\text{TD}}$ and SYCOS$_{\text{BU}}$ on different variable pairs: the original algorithm without optimizations (named as Origin), with noise theory (Noise), with MI optimization (MI Opt.), and with both noise theory and MI optimization (Both).

Figs. \ref{fig:performanceTD} and \ref{fig:performanceBU} show the results of this comparison. It can be seen that applying the proposed optimizations improves significantly the performance of both SYCOS$_{\text{TD}}$ and SYCOS$_{\text{BU}}$. Specifically, the noise pruning and the MI optimization yield different speedups on the performance of the two methods. The noise pruning improves the performance of SYCOS$_{\text{TD}}$ by a speedup of 1.8 in average, and of SYCOS$_{\text{BU}}$ by a speed up of 7.4 in average. The MI optimized computation yields an average speedup of 6.7 on SYCOS$_{\text{TD}}$, and an average speed up of 2.9 on SYCOS$_{\text{BU}}$. However, applying both optimizations always yields better speedup, with an average of 15.5 on SYCOS$_{\text{TD}}$, and 17.8 on SYCOS$_{\text{BU}}$.

In Figs. \ref{fig:performanceTD}  and \ref{fig:performanceBU}, we observe that the noise pruning is more efficient than the MI optimized computation on the pairs (Taxi, Traffic Speed), (Taxi, Rain), (Wind Direction, Wind Speed), and (Pitch Angle, Rotor Speed), and less efficient on the pairs (Wind Speed, Active Power), (Rotor Speed, Active Power), and (Wind Speed, Rotor Speed). This indicates that the effectiveness of the proposed optimizations is data dependent. However, a general trend is that, the noise pruning is more efficient than the MI optimization when the data are weakly correlated (thus, more noise is present in the data) as shown in Figs. \ref{fig:performanceBUb}, \ref{fig:performanceBUc}, and \ref{fig:performanceBUd}. In contrast, the MI optimization is more efficient when the data are strongly correlated (thus, shared MI computation occurs more often), as shown in Figs. \ref{fig:performanceTDb}, \ref{fig:performanceTDc}, and \ref{fig:performanceTDd}. This coincides with the way SYCOS$_{\text{TD}}$ and SYCOS$_{\text{BU}}$ operate, where we observe that the noise pruning is generally more efficient in SYCOS$_{\text{BU}}$ than in SYCOS$_{\text{TD}}$, and the MI optimization is more efficient in SYCOS$_{\text{TD}}$ than in SYCOS$_{\text{BU}}$.

\begin{table*}[!t]
	\centering
	\caption{Accuracy of noise pruning on SYCOS$_{\text{TD}}$ and SYCOS$_{\text{BU}}$}
	\label{tbl:accuracy}
	\vspace{-0.2cm}
	\begin{tabular}{|>{\rowmac}p{3.5cm}|>{\centering}>{\rowmac}p{0.79cm}|>{\rowmac}c|>{\rowmac}c|>{\rowmac}c|>{\rowmac}c|>{\rowmac}c|>{\rowmac}c|>{\rowmac}c|>{\rowmac}c|>{\rowmac}c|>{\rowmac}c<{\clearrow}|}
		\hline
		\multirow{2}{*}{Data} & \multirow{2}{*}{Method} & \multicolumn{10}{|c|}{Accuracy (\%)} \\ \cline{3-12}
		&  & 10K & 20K & 30K & 40K & 50K& 60K& 70K& 80K& 90K& 100K \\ \hline
		
		\multirow{2}{*}{Taxi Trips, Wind Speed} & TD &	90 &	93 &	92 &	91 &	89&	91&	89& 92 &	91& 	91 \\
		& BU &	93&	95&	91&	92&	94&	93&	93&	95&	92&	95   \\ \hline
		
		\multirow{2}{*}{Taxi Trips, Rain} & TD & 	80 &	81&	80&	80&	84&	81&	82&	85&	84&	85  \\
		& BU &	84&	85&	87&	88& 88&	88&	88&	88&	88&	84 \\ \hline
		
		\multirow{2}{*}{Taxi Fare, Rain} & TD & 100 &	85 & 100 &	100 &	93&	94&	90 &	91 &	92 &	93 \\
		& BU &	100 &	100 & 93 &	88 &	84&	88&	87&	89 &	91& 	88 \\ \hline
		
		\multirow{2}{*}{Taxi Trips, Traffic Speed} & TD &	100 &	100 &	100 &	100& 100 &	100 &	100 & 75 &	100 & 100 \\
		& BU &	98 &	98&	98 &	98 &	98 & 98 &	98&	98&	98&	98  \\ \hline
		
		\multirow{2}{*}{Wind Speed, Pitch Angle} & TD & 100&	96&	 97 &97.9& 98&	98& 100& 99&	97 & 97 \\
		& BU &	100& 98 &96 &	92 &92 &	94 & 95 &	92 & 90 &	91   \\ \hline
		
		\multirow{2}{*}{Wind Speed, Active Power} & TD & 99 & 95 & 95 &	95 &	95 & 97 &	96& 96 &96 &	96 \\
		& BU & 98 &	96& 97 &	97& 96 &	98& 100 &	94 & 95 &	94	  \\ \hline
		
		\multirow{2}{*}{Rotor Speed, Active Power} & TD & 100&	100& 100 &	100 & 100 &	100 & 100 &	100& 100&	100 \\
		& BU & 100 &	100& 100 &	100& 100 &	100 &	99& 99 &	100 & 99  \\ \hline
		
		\multirow{2}{*}{Wind Direction, Pitch Angle} & TD &100 &	100& 99 & 	88& 91 &	93& 93 &	91 & 90 &	93  \\
		& BU & 89 &	96& 95 &	92& 92 &	93& 93 &	93 &95 &	94 \\ \hline
		
		\multirow{2}{*}{Wind Direction, Wind Speed} & TD & 100&	100& 100&	100& 100 &	97 & 98 &	98 & 98 &	99  \\
		& BU & 90& 89 &	85& 85 &	87&	 90 & 88 & 89 & 85 &	92 \\ \hline
		
		\multirow{2}{*}{Wind Speed, Rotor Speed} & TD & 100&	100& 98&	96&	96 & 95 &	99&	98 & 100 &98   \\
		& BU & 99 & 98& 98 &	97& 95 &	96 & 96 &	96& 94 &	96	  \\ \hline
		
		\multirow{2}{*}{Pitch Angle, Rotor Speed} & TD & 100 & 96& 98 &	98& 95 & 	96 & 98&	97 & 97 &	96  \\
		& BU & 95 &	93& 95 & 97 & 97 &	97 &	97& 98  & 96&	98  \\ \hline
	\end{tabular}
\end{table*}

\subsubsection{The accuracy of noise pruning}
Next, we evaluate how accurate the noise pruning is compared to the exhaustive search. Particularly, we compare the similarity of windows extracted from  SYCOS$_{\text{TD}}$ and  SYCOS$_{\text{BU}}$ with and without noise pruning. 
Note that two windows are considered to be similar if they cover a similar range of data, i.e., the indices of the two windows are overlapping. 
Table \ref{tbl:accuracy} shows the accuracy of SYCOS$_{\text{TD}}$ and SYCOS$_{\text{BU}}$ on the tested datasets. It can be seen that the noise pruning yields highly accurate results, with the accuracy ranging from $80\%$ to $100\%$.

%% file: figures.tex
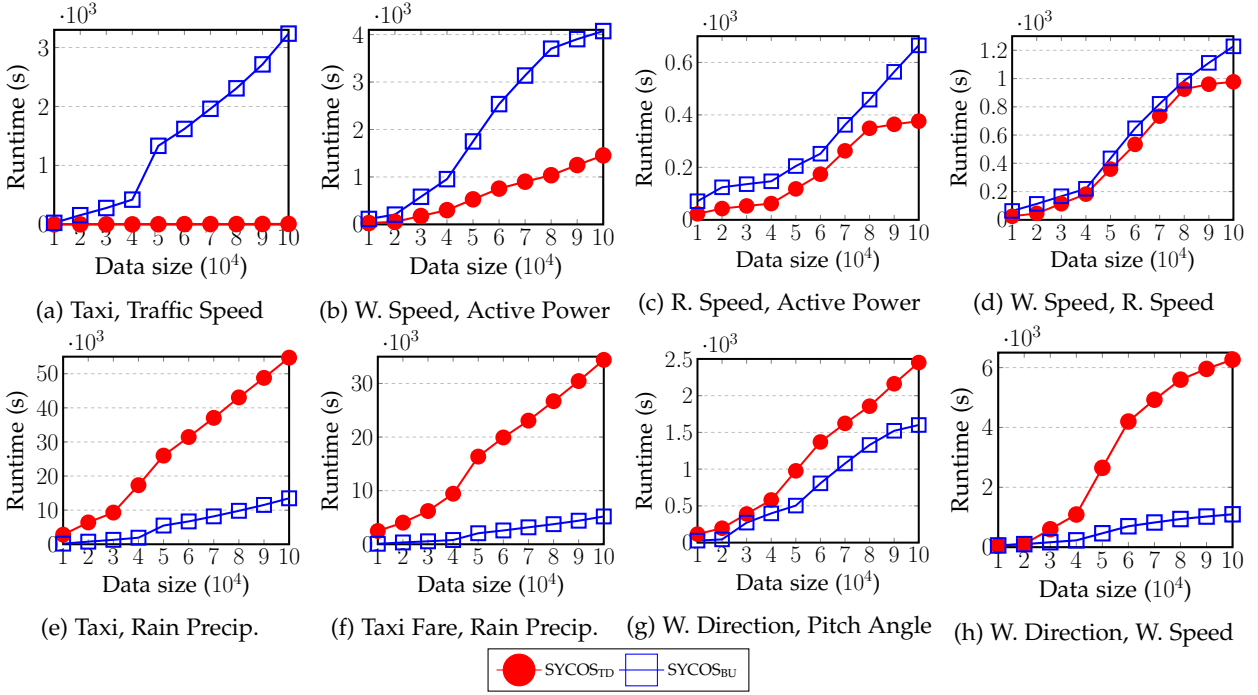
\begin{figure*}[!t]
	\vspace{-0.1in}
	\centering
	\begin{subfigure}{0.46\columnwidth}
		\centering
		\resizebox{\linewidth}{!}{
			\begin{tikzpicture}[scale=0.5]
				\begin{axis}[
					compat=newest,
					xlabel={Data size ($10^4$)},
					ylabel={Runtime (s)}, 
					label style={font=\huge},
					ticklabel style = {font=\huge},
					xmin=1, xmax=10,
					ymin=0, ymax=3300,
					xtick={1,2,3,4,5,6,7,8,9,10},
					legend columns=-1,
					legend entries = {SYCOS$_{\text{TD}}$, SYCOS$_{\text{BU}}$}, 
					legend style={nodes={scale=0.65, transform shape}},
					legend to name={legendruntime},
					scaled y ticks = base 10:-3,
					log basis y={10},
					ymajorgrids=true,
					grid style=dashed,
					line width=1.75pt
					]
					\addplot[
					color=red,
					mark=*,
					mark size=6pt,
					] 
					coordinates {
						(1,0.14)(2,0.4)(3,0.5)(4,0.7)(5,2.8)(6,4.6)(7,5)(8,5.4)(9,7)(10,8) 
					};
					\addplot[
					color=blue,
					mark=square,
					mark size=6pt,
					]
					coordinates {
						(1,26)(2,156)(3,278)(4,420)(5,1333)(6,1619)(7,1961)(8,2308)(9,2714)(10,3236) 
					};
				\end{axis}
			\end{tikzpicture}%
		}
		\caption{Taxi, Traffic Speed}
		\label{fig:runtime1}
	\end{subfigure}
	\begin{subfigure}{0.46\columnwidth}
		\centering
		\resizebox{\linewidth}{!}{
			\begin{tikzpicture}[scale=0.5]
				\begin{axis}[
					compat=newest,
					xlabel={Data size ($10^4$)},
					ylabel={Runtime (s)}, 
					label style={font=\huge},
					ticklabel style = {font=\huge},
					xmin=1, xmax=10,
					ymin=0, ymax=4100,
					xtick={1,2,3,4,5,6,7,8,9,10},
					legend columns=-1,
					legend entries = {SYCOS$_{\text{TD}}$, SYCOS$_{\text{BU}}$}, 
					legend style={nodes={scale=0.65, transform shape}},
					legend to name={legendruntime},
					scaled y ticks = base 10:-3,
					log basis y={10},
					ymajorgrids=true,
					grid style=dashed,
					line width=1.75pt
					]
					\addplot[
					color=red,
					mark=*,
					mark size=6pt,
					] 
					coordinates {
						(1,29)(2,53)(3,177)(4,300)(5,528)(6,754)(7,900)(8,1036)(9,1250)(10,1454) 
					};
					\addplot[
					color=blue,
					mark=square,
					mark size=6pt,
					]
					coordinates {
						(1,115)(2,206)(3,582)(4,954)(5,1748)(6,2535)(7,3137)(8,3703)(9,3904)(10,4074) 
					};
				\end{axis}
			\end{tikzpicture}%
		}
		\caption{W. Speed, Active Power}
		\label{fig:runtime2}
	\end{subfigure}
	\begin{subfigure}{0.46\columnwidth}
		\centering
		\resizebox{\linewidth}{!}{
			\begin{tikzpicture}[scale=0.5]
				\begin{axis}[
					compat=newest,
					xlabel={Data size ($10^4$)},
					ylabel={Runtime (s)}, 
					label style={font=\huge},
					ticklabel style = {font=\huge},
					xmin=1, xmax=10,
					ymin=0, ymax=700,
					xtick={1,2,3,4,5,6,7,8,9,10},
					legend columns=-1,
					legend entries = {SYCOS$_{\text{TD}}$, SYCOS$_{\text{BU}}$}, 
					legend style={nodes={scale=0.65, transform shape}},
					legend to name={legendruntime},
					scaled y ticks = base 10:-3,
					log basis y={10},
					ymajorgrids=true,
					grid style=dashed,
					line width=1.75pt
					]
					\addplot[
					color=red,
					mark=*,
					mark size=6pt,
					] 
					coordinates {
						(1,23)(2,43)(3,53)(4,62)(5,118)(6,174)(7,263)(8,349)(9,364)(10,376) 
					};
					\addplot[
					color=blue,
					mark=square,
					mark size=6pt,
					]
					coordinates {
						(1,71)(2,124)(3,136)(4,147)(5,205)(6,252)(7,362)(8,458)(9,564)(10,665) 
					};
				\end{axis}
			\end{tikzpicture}%
		}
		\caption{R. Speed, Active Power}
		\label{fig:runtime3}
	\end{subfigure}
	\begin{subfigure}{0.46\columnwidth}
		\centering
		\resizebox{\linewidth}{!}{
			\begin{tikzpicture}[scale=0.5]
				\begin{axis}[
					compat=newest,
					xlabel={Data size ($10^4$)},
					ylabel={Runtime (s)}, 
					label style={font=\huge},
					ticklabel style = {font=\huge},
					xmin=1, xmax=10,
					ymin=0, ymax=1300,
					xtick={1,2,3,4,5,6,7,8,9,10},
					legend columns=-1,
					legend entries = {SYCOS$_{\text{TD}}$, SYCOS$_{\text{BU}}$}, 
					legend style={nodes={scale=0.65, transform shape}},
					legend to name={legendruntime},
					scaled y ticks = base 10:-3,
					log basis y={10},
					ymajorgrids=true,
					grid style=dashed,
					line width=1.75pt
					]
					\addplot[
					color=red,
					mark=*,
					mark size=6pt,
					] 
					coordinates {
						(1,25)(2,46)(3,114)(4,182)(5,359)(6,534)(7,734)(8,925)(9,960)(10,977) 
					};
					\addplot[
					color=blue,
					mark=square,
					mark size=6pt,
					]
					coordinates {
						(1,63)(2,113)(3,167)(4,220)(5,435)(6,648)(7,821)(8,985)(9,1111)(10,1228) 
					};
				\end{axis}
			\end{tikzpicture}%
		}
		\caption{W. Speed, R. Speed}
		\label{fig:runtime4}
	\end{subfigure}
	\begin{subfigure}{0.46\columnwidth}
		\centering
		\resizebox{\linewidth}{!}{
			\begin{tikzpicture}[scale=0.5]
				\begin{axis}[
					compat=newest,
					xlabel={Data size ($10^4$)},
					ylabel={Runtime (s)}, 
					label style={font=\huge},
					ticklabel style = {font=\huge},
					xmin=1, xmax=10,
					ymin=0, ymax=55000,
					xtick={1,2,3,4,5,6,7,8,9,10},
					legend columns=-1,
					legend entries = {SYCOS$_{\text{TD}}$, SYCOS$_{\text{BU}}$}, 
					legend style={nodes={scale=0.65, transform shape}},
					legend to name={legendruntime},
					scaled y ticks = base 10:-3,
					log basis y={10},
					ymajorgrids=true,
					grid style=dashed,
					line width=1.75pt
					]
					\addplot[
					color=red,
					mark=*,
					mark size=6pt,
					] 
					coordinates {
						(1,2810)(2,6399)(3,9287)(4,17318)(5,25992)(6,31443)(7,37063)(8,43060)(9,48784)(10,54714) 
					};
					\addplot[
					color=blue,
					mark=square,
					mark size=6pt,
					]
					coordinates {
						(1,179)(2,759)(3,1297)(4,1922)(5,5481)(6,6697)(7,8183)(8,9805)(9,11531)(10,13474) 
					};
				\end{axis}
			\end{tikzpicture}%
		}
		\caption{Taxi, Rain Precip.}
		\label{fig:runtime5}
	\end{subfigure}
	\begin{subfigure}{0.46\columnwidth}
		\centering
		\resizebox{\linewidth}{!}{
			\begin{tikzpicture}[scale=0.5]
				\begin{axis}[
					compat=newest,
					xlabel={Data size ($10^4$)},
					ylabel={Runtime (s)}, 
					label style={font=\huge},
					ticklabel style = {font=\huge},
					xmin=1, xmax=10,
					ymin=0, ymax=35000,
					xtick={1,2,3,4,5,6,7,8,9,10},
					legend columns=-1,
					legend entries = {SYCOS$_{\text{TD}}$, SYCOS$_{\text{BU}}$}, 
					legend style={nodes={scale=0.65, transform shape}},
					legend to name={legendruntime},
					scaled y ticks = base 10:-3,
					log basis y={10},
					ymajorgrids=true,
					grid style=dashed,
					line width=1.75pt
					]
					\addplot[
					color=red,
					mark=*,
					mark size=6pt,
					] 
					coordinates {
						(1,2473)(2,4035)(3,6196)(4,9429)(5,16353)(6,19923)(7,23068)(8,26697)(9,30448)(10,34403) 
					};
					\addplot[
					color=blue,
					mark=square,
					mark size=6pt,
					]
					coordinates {
						(1,93)(2,336)(3,549)(4,795)(5,2055)(6,2576)(7,3161)(8,3750)(9,4393)(10,5188) 
					};
				\end{axis}
			\end{tikzpicture}%
		}
		\caption{Taxi Fare, Rain Precip.}
		\label{fig:runtime6}
	\end{subfigure}
	\begin{subfigure}{0.46\columnwidth}
		\centering
		\resizebox{\linewidth}{!}{
			\begin{tikzpicture}[scale=0.5]
				\begin{axis}[
					compat=newest,
					xlabel={Data size ($10^4$)},
					ylabel={Runtime (s)}, 
					label style={font=\huge},
					ticklabel style = {font=\huge},
					xmin=1, xmax=10,
					ymin=0, ymax=2500,
					xtick={1,2,3,4,5,6,7,8,9,10},
					legend columns=-1,
					legend entries = {SYCOS$_{\text{TD}}$, SYCOS$_{\text{BU}}$}, 
					legend style={nodes={scale=0.65, transform shape}},
					legend to name={legendruntime},
					scaled y ticks = base 10:-3,
					log basis y={10},
					ymajorgrids=true,
					grid style=dashed,
					line width=1.75pt
					]
					\addplot[
					color=red,
					mark=*,
					mark size=6pt,
					] 
					coordinates {
						(1,116)(2,196)(3,390)(4,581)(5,977)(6,1370)(7,1623)(8,1857)(9,2162)(10,2451) 
					};
					\addplot[
					color=blue,
					mark=square,
					mark size=6pt,
					]
					coordinates {
						(1,27)(2,44)(3,272)(4,398)(5,504)(6,807)(7,1077)(8,1329)(9,1522)(10,1601) 
					};
				\end{axis}
			\end{tikzpicture}%
		}
		\caption{W. Direction, Pitch Angle}
		\label{fig:runtime7}
	\end{subfigure}
	\begin{subfigure}{0.46\columnwidth}
		\centering
		\resizebox{\linewidth}{!}{
			\begin{tikzpicture}[scale=0.5]
				\begin{axis}[
					compat=newest,
					xlabel={Data size ($10^4$)},
					ylabel={Runtime (s)}, 
					label style={font=\huge},
					ticklabel style = {font=\huge},
					xmin=1, xmax=10,
					ymin=0, ymax=6500,
					xtick={1,2,3,4,5,6,7,8,9,10},
					legend columns=-1,
					legend entries = {SYCOS$_{\text{TD}}$, SYCOS$_{\text{BU}}$}, 
					legend style={nodes={scale=0.65, transform shape}},
					legend to name={legendruntime},
					scaled y ticks = base 10:-3,
					log basis y={10},
					ymajorgrids=true,
					grid style=dashed,
					line width=1.75pt
					]
					\addplot[
					color=red,
					mark=*,
					mark size=6pt,
					] 
					coordinates {
						(1,56)(2,101)(3,599)(4,1091)(5,2650)(6,4198)(7,4927)(8,5597)(9,5956)(10,6268) 
					};
					\addplot[
					color=blue,
					mark=square,
					mark size=6pt,
					]
					coordinates {
						(1,58)(2,101)(3,165)(4,228)(5,464)(6,698)(7,824)(8,940)(9,1023)(10,1098) 
					};
				\end{axis}
			\end{tikzpicture}%
		}
		\caption{W. Direction, W. Speed}
		\label{fig:runtime8}
	\end{subfigure}
	\ref{legendruntime}
	\vspace{-0.07in}
	\caption{Runtime comparison between SYCOS$_{\text{TD}}$ and SYCOS$_{\text{BU}}$}
	\label{fig:runtime}
\end{figure*} 

\begin{figure*}[!t]
	\centering
	\vspace{-0.2in}
	\begin{subfigure}{0.46\columnwidth}
		\centering
		\resizebox{\linewidth}{!}{
			\begin{tikzpicture}[scale=0.5]
				\begin{axis}[
					compat=newest,
					xlabel={Data size ($10^3$)},
					ylabel={Runtime (s)}, 
					label style={font=\huge},
					ticklabel style = {font=\huge},
					xmin=1, xmax=100,
					ymin=0, ymax=10,
					xtick={1,20,40,60,80,100},
					legend columns=-1,
					legend entries = {Origin, Noise, MI Opt., Both}, 
					legend style={nodes={scale=0.65, transform shape}},
					legend to name={legendpruning3},
					log basis y={10},
					ymajorgrids=true,
					grid style=dashed,
					line width=1.75pt
					]
					\addplot[
					color=red,
					mark=*,
					mark size=6pt,
					] 
					coordinates {
						(1,0.009796)(20,0.404064)(40,0.743233)(60,4.69827)(80,5.46861)(100,8.94321)
					};
					\addplot[
					color=blue,
					mark=square,
					mark size=6pt,
					]
					coordinates {
						(1,0.009414)(20,0.28)(40,0.59)(60,3.64)(80,4.54)(100,7.54)
					};
					\addplot[
					color=violet,
					mark=diamond,
					mark size=6pt,
					] 
					coordinates {
						(1,0.010492)(20,1.00382)(40,1.22643)(60,2.7)(80,3)(100,3.3)
					};
					
					\addplot[
					color=green!20!black,
					mark=asterisk,
					mark size=6pt,
					] 
					coordinates {
						(1,0.01)(20,0.15)(40,0.29)(60,0.52)(80,0.68)(100,0.96)
					};
				\end{axis}
			\end{tikzpicture}%
		}
		\caption{Taxi, Traffic Speed}
		\label{fig:performanceTDa}
	\end{subfigure}
	\begin{subfigure}{0.46\columnwidth}
		\centering
		\resizebox{\linewidth}{!}{
			\begin{tikzpicture}[scale=0.5]
				\begin{axis}[
					compat=newest,
					xlabel={Data size ($10^3$)},
					ylabel={Runtime (s)}, 
					label style={font=\huge},
					ticklabel style = {font=\huge},
					xmin=1, xmax=100,
					ymin=0, ymax=1500,
					xtick={1,20,40,60,80,100},
					legend columns=-1,
					legend entries = {Origin, Noise, MI Opt., Both}, 
					legend style={nodes={scale=0.65, transform shape}},
					legend to name={legendpruning3},
					scaled y ticks = base 10:-3,
					log basis y={10},
					ymajorgrids=true,
					grid style=dashed,
					line width=1.75pt
					]
					\addplot[
					color=red,
					mark=*,
					mark size=6pt,
					] 
					coordinates {
						(1,0.004)(20,53)(40,300)(60,754)(80,1036)(100,1454)
					};
					\addplot[
					color=violet,
					mark=diamond,
					mark size=6pt,
					] 
					coordinates {
						(1,0.004)(20,46)(40,285)(60,533)(80,853)(100,927)
					};
					\addplot[
					color=blue,
					mark=square,
					mark size=6pt,
					]
					coordinates {
						(1,0.003)(20,19)(40,115)(60,165)(80,234)(100,276)
					};				
					\addplot[
					color=green!20!black,
					mark=asterisk,
					mark size=6pt,
					] 
					coordinates {
						(1,0.007)(20,15)(40,91)(60,132)(80,177)(100,213)
					};
				\end{axis}
			\end{tikzpicture}%
		}
		\caption{W. Speed, Active Power}
		\label{fig:performanceTDb}
	\end{subfigure}
	\begin{subfigure}{0.46\columnwidth}
		\centering
		\resizebox{\linewidth}{!}{
			\begin{tikzpicture}[scale=0.5]
				\begin{axis}[
					compat=newest,
					xlabel={Data size ($10^3$)},
					ylabel={Runtime (s)}, 
					label style={font=\huge},
					ticklabel style = {font=\huge},
					xmin=1, xmax=100,
					ymin=0, ymax=400,
					xtick={1,20,40,60,80,100},
					legend columns=-1,
					legend entries = {Origin, Noise, MI Opt., Both}, 
					legend style={nodes={scale=0.65, transform shape}},
					legend to name={legendpruning3},
					scaled y ticks = base 10:-3,
					log basis y={10},
					ymajorgrids=true,
					grid style=dashed,
					line width=1.75pt
					]
					\addplot[
					color=red,
					mark=*,
					mark size=6pt,
					] 
					coordinates {
						(1,0.008)(20,43)(40,62)(60,174)(80,349)(100,376)
					};
					\addplot[
					color=violet,
					mark=diamond,
					mark size=6pt,
					] 
					coordinates {
						(1,0.006)(20,26)(40,49)(60,56)(80,95)(100,129)
					};
					\addplot[
					color=blue,
					mark=square,
					mark size=6pt,
					]				
					coordinates {
						(1,0.005)(20,25)(40,21)(60,31)(80,55)(100,59)
					};
					\addplot[
					color=green!20!black,
					mark=asterisk,
					mark size=6pt,
					] 
					coordinates {
						(1,0.009)(20,10)(40,10)(60,16)(80,30)(100,38.15)
					};
				\end{axis}
			\end{tikzpicture}%
		}
		\caption{R. Speed, Active Power}
		\label{fig:performanceTDc}
	\end{subfigure}
	\begin{subfigure}{0.46\columnwidth}
		\centering
		\resizebox{\linewidth}{!}{
			\begin{tikzpicture}[scale=0.5]
				\begin{axis}[
					compat=newest,
					xlabel={Data size ($10^3$)},
					ylabel={Runtime (s)}, 
					label style={font=\huge},
					ticklabel style = {font=\huge},
					xmin=1, xmax=100,
					ymin=0, ymax=1000,
					xtick={1,20,40,60,80,100},
					legend columns=-1,
					legend entries = {Origin, Noise, MI Opt., Both}, 
					legend style={nodes={scale=0.65, transform shape}},
					legend to name={legendpruning3},
					scaled y ticks = base 10:-3,
					log basis y={10},
					ymajorgrids=true,
					grid style=dashed,
					line width=1.75pt
					]
					\addplot[
					color=red,
					mark=*,
					mark size=6pt,
					] 
					coordinates {
						(1,0.007)(20,46)(40,182)(60,534)(80,925)(100,967)
					};
					\addplot[
					color=violet,
					mark=diamond,
					mark size=6pt,
					] 
					coordinates {
						(1,0.006)(20,36)(40,174)(60,449)(80,608)(100,781)
					};
					\addplot[
					color=blue,
					mark=square,
					mark size=6pt,
					]				
					coordinates {
						(1,0.005)(20,16)(40,33)(60,40)(80,68)(100,71)
					};
					\addplot[
					color=green!20!black,
					mark=asterisk,
					mark size=6pt,
					] 
					coordinates {
						(1,0.002)(20,8.4)(40,25)(60,30)(80,52)(100,62)
					};
				\end{axis}
			\end{tikzpicture}%
		}
		\caption{W. Speed, R. Speed}
		\label{fig:performanceTDd}
	\end{subfigure}
	\ref{legendpruning3}
	\vspace{-0.07in}
	\caption{Impact of Noise Pruning and MI Optimization on SYCOS$_{\text{TD}}$ performance}
	\label{fig:performanceTD}
\end{figure*}
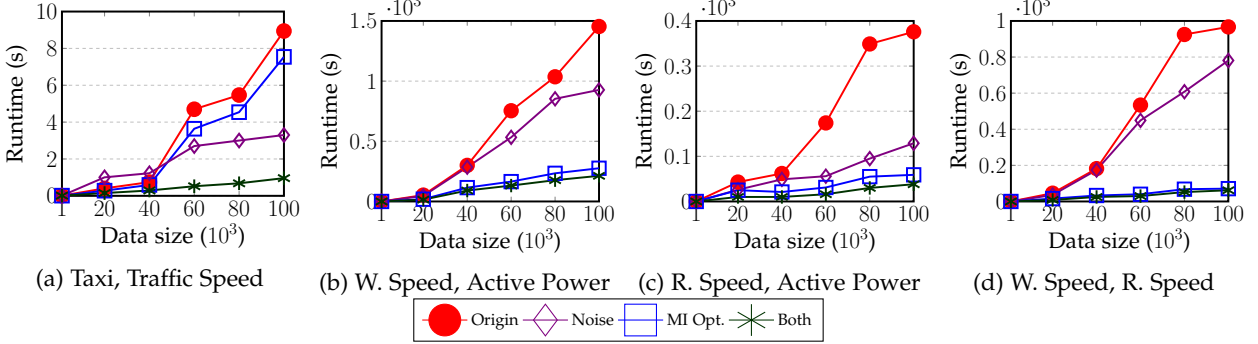 

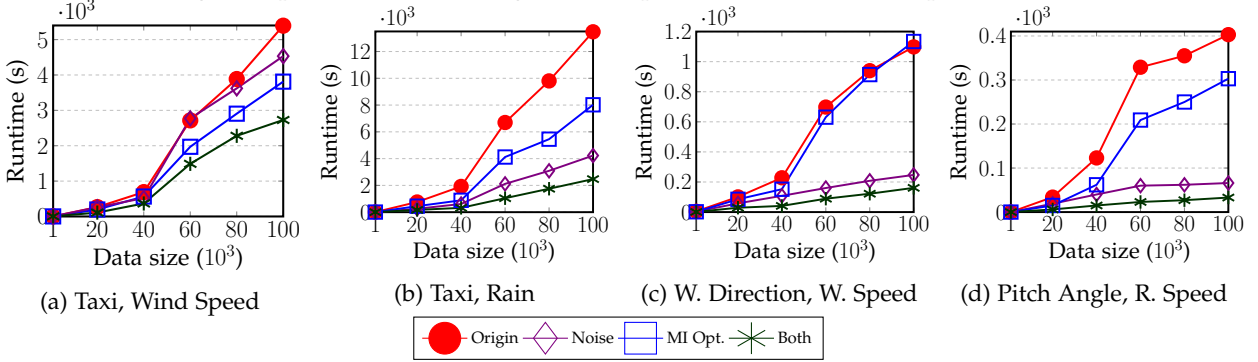
\begin{figure*}[!t]
	\centering
	\vspace{-0.2in}
	\begin{subfigure}{0.46\columnwidth}
		\centering
		\resizebox{\linewidth}{!}{
			\begin{tikzpicture}[scale=0.5]
				\begin{axis}[
					compat=newest,
					xlabel={Data size ($10^3$)},
					ylabel={Runtime (s)}, 
					label style={font=\huge},
					ticklabel style = {font=\huge},
					xmin=1, xmax=100,
					ymin=0, ymax=5400,
					xtick={1,20,40,60,80,100},
					legend columns=-1,
					legend entries = {Origin, Noise, MI Opt., Both}, 
					legend style={nodes={scale=0.65, transform shape}},
					legend to name={legendpruning4},
					scaled y ticks = base 10:-3,
					log basis y={10},
					ymajorgrids=true,
					grid style=dashed,
					line width=1.75pt
					]
					\addplot[
					color=red,
					mark=*,
					mark size=6pt,
					] 
					coordinates {
						(1,2)(20,265)(40,689)(60,2716)(80,3886)(100,5394)
					};
					\addplot[
					color=blue,
					mark=square,
					mark size=6pt,
					]
					coordinates {
						(1,2)(20,197)(40,563)(60,1970)(80,2906)(100,3812)
					};
					\addplot[
					color=violet,
					mark=diamond,
					mark size=6pt,
					] 
					coordinates {
						(1,3)(20,258)(40,520)(60,2772)(80,3622)(100,4533)
					};
					
					\addplot[
					color=green!20!black,
					mark=asterisk,
					mark size=6pt,
					] 
					coordinates {
						(1,3)(20,110)(40,367)(60,1488)(80,2280)(100,2728)
					};
				\end{axis}
			\end{tikzpicture}%
		}
		\caption{Taxi, Wind Speed}
		\label{fig:performanceBUa}
	\end{subfigure}
	\begin{subfigure}{0.46\columnwidth}
		\centering
		\resizebox{\linewidth}{!}{
			\begin{tikzpicture}[scale=0.5]
				\begin{axis}[
					compat=newest,
					xlabel={Data size ($10^3$)},
					ylabel={Runtime (s)}, 
					label style={font=\huge},
					ticklabel style = {font=\huge},
					xmin=1, xmax=100,
					ymin=0, ymax=13500,
					xtick={1,20,40,60,80,100},
					legend columns=-1,
					legend entries = {Origin, Noise, MI Opt., Both}, 
					legend style={nodes={scale=0.65, transform shape}},
					legend to name={legendpruning4},
					scaled y ticks = base 10:-3,
					log basis y={10},
					ymajorgrids=true,
					grid style=dashed,
					line width=1.75pt
					]
					\addplot[
					color=red,
					mark=*,
					mark size=6pt,
					] 
					coordinates {
						(1,8)(20,759)(40,1922)(60,6697)(80,9805)(100,13474)
					};
					\addplot[
					color=violet,
					mark=diamond,
					mark size=6pt,
					] 
					coordinates {
						(1,3)(20,252)(40,612)(60,2114)(80,3089)(100,4229)
					};	
					
					\addplot[
					color=blue,
					mark=square,
					mark size=6pt,
					]
					coordinates {
						(1,7)(20,447)(40,886)(60,4119)(80,5452)(100,8020)
					};
					\addplot[
					color=green!20!black,
					mark=asterisk,
					mark size=6pt,
					] 
					coordinates {
						(1,2)(20,169)(40,322)(60,1043)(80,1750)(100,2470)
					};	
				\end{axis}
			\end{tikzpicture}
		}
		\caption{Taxi, Rain}
		\label{fig:performanceBUb}
	\end{subfigure}
	\begin{subfigure}{0.46\columnwidth}
		\centering
		\resizebox{\linewidth}{!}{
			\begin{tikzpicture}[scale=0.5]
				\begin{axis}[
					compat=newest,
					xlabel={Data size ($10^3$)},
					ylabel={Runtime (s)}, 
					label style={font=\huge},
					ticklabel style = {font=\huge},
					xmin=1, xmax=100,
					ymin=0, ymax=1200,
					xtick={1,20,40,60,80,100},
					legend columns=-1,
					legend entries = {Origin, Noise, MI Opt., Both}, 
					legend style={nodes={scale=0.65, transform shape}},
					legend to name={legendpruning4},
					scaled y ticks = base 10:-3,
					log basis y={10},
					ymajorgrids=true,
					grid style=dashed,
					line width=1.75pt
					]
					\addplot[
					color=red,
					mark=*,
					mark size=6pt,
					] 					
					coordinates {
						(1,5)(20,101)(40,228)(60,698)(80,940)(100,1098)
					};
					\addplot[
					color=violet,
					mark=diamond,
					mark size=6pt,
					] 					
					coordinates {
						(1,3)(20,58)(40,108)(60,160)(80,208)(100,247)
					};
					\addplot[
					color=blue,
					mark=square,
					mark size=6pt,
					]				
					coordinates {
						(1,2)(20,85)(40,152)(60,631)(80,914)(100,1134)
					};
					\addplot[
					color=green!20!black,
					mark=asterisk,
					mark size=6pt,
					] 
					coordinates {
						(1,1)(20,29)(40,41)(60,89)(80,122)(100,161)
					};
				\end{axis}
			\end{tikzpicture}%
		}
		\caption{W. Direction, W. Speed}
		\label{fig:performanceBUc}
	\end{subfigure}
	\begin{subfigure}{0.46\columnwidth}
		\centering
		\resizebox{\linewidth}{!}{
			\begin{tikzpicture}[scale=0.5]
				\begin{axis}[
					compat=newest,
					xlabel={Data size ($10^3$)},
					ylabel={Runtime (s)}, 
					label style={font=\huge},
					ticklabel style = {font=\huge},
					xmin=1, xmax=100,
					ymin=0, ymax=410,
					xtick={1,20,40,60,80,100},
					legend columns=-1,
					legend entries = {Origin, Noise, MI Opt., Both}, 
					legend style={nodes={scale=0.65, transform shape}},
					legend to name={legendpruning4},
					scaled y ticks = base 10:-3,
					log basis y={10},
					ymajorgrids=true,
					grid style=dashed,
					line width=1.75pt
					]
					\addplot[
					color=red,
					mark=*,
					mark size=6pt,
					] 
					coordinates {
						(1,0.004)(20,34)(40,123)(60,329)(80,355)(100,403)
					};
					\addplot[
					color=violet,
					mark=diamond,
					mark size=6pt,
					] 
					coordinates {
						(1,0.004)(20,19)(40,40)(60,60)(80,62)(100,66)
					};
					\addplot[
					color=blue,
					mark=square,
					mark size=6pt,
					]				
					coordinates {
						(1,0.004)(20,15)(40,62)(60,209)(80,250)(100,303)
					};
					\addplot[
					color=green!20!black,
					mark=asterisk,
					mark size=6pt,
					] 
					coordinates {
						(1,0.007)(20,6)(40,15)(60,23)(80,27)(100,33)
					};
				\end{axis}
			\end{tikzpicture}%
		}
		\caption{Pitch Angle, R. Speed}
		\label{fig:performanceBUd}
	\end{subfigure}
	\ref{legendpruning4}
	\vspace{-0.07in}
	\caption{Impact of Noise Pruning and MI Optimization on SYCOS$_{\text{BU}}$ performance}
	\label{fig:performanceBU}
\end{figure*} 

\begin{figure*}[!t]
	\centering
	\begin{subfigure}{0.46\columnwidth}
		\centering
		\resizebox{\linewidth}{!}{
			\pgfplotsset{
				tick label style={font=\small},
			} 
			\begin{tikzpicture}
				\pgfplotsset{lineplot/.style={blue,mark=*,sharp plot,line legend}}
				\begin{axis}[
					ybar,
					ylabel={ Runtimes (hours)},
					xlabel={Number of Nodes},
					ymin=0,
					symbolic x coords={1,2,4,6,8,10,12,14,16,18,20},
					xtick=data,
					ylabel near ticks,
					nodes near coords,
					every node near coord/.append style={font=\footnotesize},
					nodes near coords align={vertical},
					bar width=3mm,
					legend style={at={(0.5,0.8)},anchor=south},
					]
					\legend{Speed up, Runtimes}
					\addlegendimage{customlegend}
					\addplot[ybar, nodes near coords, fill=blue!20] coordinates {(1,2.2) (2,1.1) (4,0.56) (6,0.38) (8,0.28) (10,0.23) (12,0.19) (14,0.17) (16,0.14) (18,0.12) (20,0.115)};
				\end{axis}
				\begin{axis}[
					styleline, 
					]
					\addplot[draw=blue,lineplot] 
					coordinates {(1,1) (2,1.996) (4,3.88) (6,5.73) (8,7.62) (10,9.51) (12,11.42) (14,13.31) (16,15.2) (18,17.11) (20,19.01)};
				\end{axis}
		\end{tikzpicture}
		}
		\vspace{-0.2in}
		\caption{Taxi, Traffic Speed}
		\label{fig:scalabilitytest1}
	\end{subfigure}
	\begin{subfigure}{0.46\columnwidth}
		\centering
		\resizebox{\linewidth}{!}{
			
			\pgfplotsset{
				tick label style={font=\small},
			} 
			\begin{tikzpicture}
				\pgfplotsset{lineplot/.style={blue,mark=*,sharp plot,line legend}}
				\begin{axis}[
					ybar,
					ylabel={ Runtimes (hours)},
					xlabel={Number of Nodes},
					ymin=0,
					symbolic x coords={1,2,4,6,8,10,12,14,16,18,20},
					xtick=data,
					ylabel near ticks,
					nodes near coords,
					every node near coord/.append style={font=\footnotesize},
					nodes near coords align={vertical},
					bar width=3mm,
					legend style={at={(0.5,0.8)},anchor=south},
					]
					\legend{Speed up, Runtimes}
					\addlegendimage{customlegend}
					\addplot[ybar, nodes near coords, fill=blue!20] coordinates {(1,50.48) (2,25.9) (4,13.3) (6,9) (8,6.8) (10,5.5) (12,4.5) (14,3.9) (16,3.4) (18,3) (20,2.7)};
				\end{axis}
				\begin{axis}[
					xmin = 1,
					xmax = 20,
					axis y line=right,
					axis x line=none,
					symbolic x coords={1,2,4,6,8,10,12,14,16,18,20},
					xtick=data,
					enlarge x limits=0.1,
					ylabel = {Speed up},
					ylabel near ticks,
					ylabel style={rotate=-180},
					xmajorgrids,
					scaled y ticks = false,
					ymin=0, ymax=20,
					]
					\addplot[draw=blue,lineplot] 
					coordinates {(1,1) (2,1.94) (4,3.77) (6,5.57) (8,7.4) (10,9.26) (12,11.11) (14,12.95) (16,14.8) (18,16.64) (20,18.49)};
				\end{axis}
			\end{tikzpicture}
		}
		\vspace{-0.2in}
		\caption{W. Speed, Active Power}
		\label{fig:scalabilitytest2}
	\end{subfigure}
	\begin{subfigure}{0.46\columnwidth}
		\centering
		\resizebox{\linewidth}{!}{
			
			\pgfplotsset{
				tick label style={font=\small},
			} 
			\begin{tikzpicture}
				\pgfplotsset{lineplot/.style={blue,mark=*,sharp plot,line legend}}
				\begin{axis}[
					ybar,
					ylabel={ Runtimes (hours)},
					xlabel={Number of Nodes},
					ymin=0,
					symbolic x coords={1,2,4,6,8,10,12,14,16,18,20},
					xtick=data,
					ylabel near ticks,
					nodes near coords,
					every node near coord/.append style={font=\footnotesize},
					nodes near coords align={vertical},
					bar width=3mm,
					legend style={at={(0.5,0.8)},anchor=south},
					]
					\legend{Speed up, Runtimes}
					\addlegendimage{customlegend}
					\addplot[ybar, nodes near coords, fill=blue!20] coordinates {(1,187) (2,98) (4,51) (6,34.2) (8,25.7) (10,21) (12,17.2) (14,14.7) (16,12.9) (18,11.5) (20,10.3)};
				\end{axis}
				\begin{axis}[
					xmin = 1,
					xmax = 20,
					axis y line=right,
					axis x line=none,
					symbolic x coords={1,2,4,6,8,10,12,14,16,18,20},
					xtick=data,
					enlarge x limits=0.1,
					ylabel = {Speed up},
					ylabel near ticks,
					ylabel style={rotate=-180},
					xmajorgrids,
					scaled y ticks = false,
					ymin=0, ymax=20,
					]
					\addplot[draw=blue,lineplot] 
					coordinates {(1,1) (2,1.9) (4,3.69) (6,5.47) (8,7.27) (10,9.08) (12,10.897) (14,12.7) (16,14.52) (18,16.33) (20,18.14)};
				\end{axis}
			\end{tikzpicture}
		}
		\vspace{-0.2in}
		\caption{Taxi, Rain Precip.}
		\label{fig:scalabilitytest3}
	\end{subfigure}
	\begin{subfigure}{0.46\columnwidth}
		\centering
		\resizebox{\linewidth}{!}{
			
			\pgfplotsset{
				tick label style={font=\small},
			} 
			\begin{tikzpicture}
				\pgfplotsset{lineplot/.style={blue,mark=*,sharp plot,line legend}}
				\begin{axis}[
					ybar,
					ylabel={ Runtimes (hours)},
					xlabel={Number of Nodes},
					ymin=0,
					symbolic x coords={1,2,4,6,8,10,12,14,16,18,20},
					xtick=data,
					ylabel near ticks,
					nodes near coords,
					every node near coord/.append style={font=\footnotesize},
					nodes near coords align={vertical},
					bar width=3mm,
					legend style={at={(0.5,0.8)},anchor=south},
					]
					\legend{Speed up, Runtimes}
					\addlegendimage{customlegend}
					\addplot[ybar, nodes near coords, fill=blue!20]  coordinates {(1,140) (2,71) (4,36) (6,25) (8,18.5) (10,14.8) (12,12.3) (14,10.6) (16,9.2) (18,8.2) (20,7.4)};
				\end{axis}
				\begin{axis}[
					xmin = 1,
					xmax = 20,
					axis y line=right,
					axis x line=none,
					symbolic x coords={1,2,4,6,8,10,12,14,16,18,20},
					xtick=data,
					enlarge x limits=0.1,
					ylabel = {Speed up},
					ylabel near ticks,
					ylabel style={rotate=-180},
					xmajorgrids,
					scaled y ticks = false,
					ymin=0, ymax=20,
					]
					\addplot[draw=blue,lineplot] 
					coordinates {(1,1) (2,1.978) (4,3.84) (6,5.68) (8,7.5) (10,9.4) (12,11.3) (14,13.2) (16,15.1) (18,16.95) (20,18.8)};
				\end{axis}
			\end{tikzpicture}
		}
		\vspace{-0.2in}
		\caption{Pitch Angle, Rotor Speed}
		\label{fig:scalabilitytest4}
	\end{subfigure}
	\caption{iSYCOS scalability evaluation} 
	\label{fig:scalabilitytest}
\end{figure*}
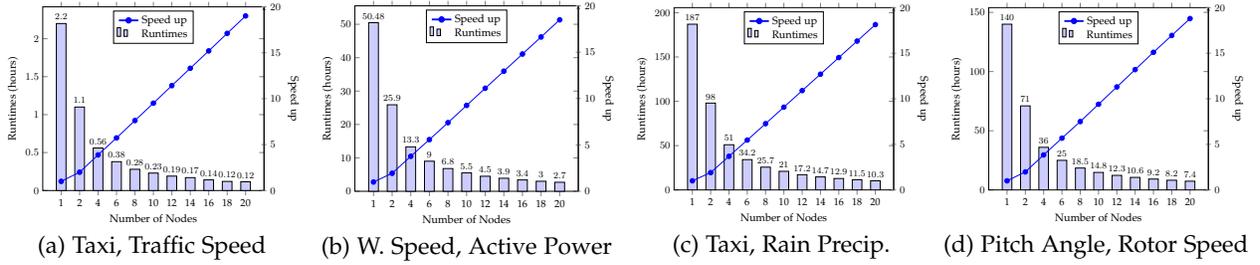

%% file: selectionmethodevaluation.tex
\subsubsection{Effectiveness of the selection algorithm}
We evaluate the effectiveness of iSYCOS in selecting the most efficient search method for a given pair of input time series. More specifically, for each pair of variables, we perform random sampling to select $6$ data partitions, and apply the selection algorithm (Section \ref{sec:selectionanalysis}) on them. Table \ref{tbl:selection} ($\alpha = 0.5$) shows the computed efficiency scores of SYCOS$_{\text{TD}}$ and SYCOS$_{\text{BU}}$. 
The method with higher score (in bold) is chosen. As can be seen, iSYCOS selects SYCOS$_{\text{TD}}$ for the pairs (Taxi Trips, Traffic Speed), (Wind Speed, Active Power), (Rotor Speed, Active Power), and SYCOS$_{\text{BU}}$ for the rest. Compared to the analysis in Section \ref{sec:runtime}, iSYCOS has correctly selected the most efficient search approach for each tested variable pair. 

Notice however that iSYCOS has selected SYCOS$_{\text{BU}}$ instead of SYCOS$_{\text{TD}}$ for the pair (Wind Speed, Rotor Speed). When we further analyze this pair, we observe that the runtimes of SYCOS$_{\text{TD}}$ or SYCOS$_{\text{BU}}$ on this pair are similar (Fig. \ref{fig:runtime4} on the sampled data partitions. Thus, their efficiency scores are close. However, as we change the weight $\alpha$ to $\alpha = 0.2$ (put more weight on the large window than on the runtime), we obtain the efficiency scores $\text{nscore}_{\text{BU}} = 0.79$ and $\text{nscore}_{\text{TD}} = 0.86$. In this case, iSYCOS selects SYCOS$_{\text{TD}}$. 

Next, we analyze the impact of $\alpha$ on iSYCOS selection by changing $\alpha$ in the range $(0.2,...,0.8)$, and computing the efficiency scores of SYCOS$_{\text{TD}}$ and SYCOS$_{\text{BU}}$ accordingly. We find that when the data are strongly or weakly correlated,  iSYCOS selection is robust w.r.t. $\alpha$: changing $\alpha$ does not change the selection. This is intuitive as when SYCOS$_{\text{TD}}$ is significantly better than SYCOS$_{\text{BU}}$ (strong correlation) or vice versa (weak correlation), their efficiency scores cannot be changed by $\alpha$ in a way that alters the selection. However, in cases where data are moderately correlated, and SYCOS$_{\text{TD}}$ and SYCOS$_{\text{BU}}$ performances become comparable such as the case of (Wind Speed, Rotor Speed), changing $\alpha$ can change the selection between SYCOS$_{\text{TD}}$ and SYCOS$_{\text{BU}}$, depending on whether the runtime or the window size is weighted more.

\begin{table}[!t]
	\footnotesize
	\caption{Efficiency scores of SYCOS$_{\text{TD}}$ or SYCOS$_{\text{BU}}$}
	\label{tbl:selection}
	\vspace{-0.1in}
	\centering
	\begin{tabular}{|>{\rowmac}p{4.5cm}|>{\rowmac}c|>{\rowmac}c<{\clearrow}|}
		\hline 
		Data & SYCOS$_{\text{TD}}$ & SYCOS$_{\text{BU}}$ \\ \hline
		Taxi Trips, Traffic Speed & \bfseries{5.86} & 0.74 \\ \hline
		Wind Speed, Active Power &  \bfseries{1.48} & 1.01 \\ \hline
		Rotor Speed, Active Power &  \bfseries{2.5} & 0.66 \\ \hline
		Wind Speed, Rotor Speed &  1.14 & \bfseries{1.46} \\ \hline
		Taxi Trips, Rain& 0.62  & \setrow{\bfseries}3.02   \\ \hline
		Taxi Fare, Rain& 0.54 & \setrow{\bfseries}41.21  \\ \hline
		Wind Direction, Pitch Angle & 1.16& \setrow{\bfseries}1.34 \\ \hline
		Wind Direction, Wind Speed & 0.8 & \setrow{\bfseries}2.7 \\ \hline
	\end{tabular}
\end{table}

%% file: scalability.tex
\subsubsection{iSYCOS scalability}
Finally, we evaluate the scalability of iSYCOS on large synthetic datasets generated from the smart city and the energy data. The synthetic data contain up to 100 millions of data points.  We use a Spark cluster containing 20 high-end nodes, each with 64 cores, 256GB of RAM, 24TB of storage, and runs Apache Spark 3.1. 

On the Spark cluster, we fix the data size to 100M data points while changing the number of worker nodes from 1 to 20. Fig. \ref{fig:scalabilitytest} shows the runtimes of the distributed iSYCOS on different datasets, as well as the speedup obtained when the computing resources increase. It can be seen that, iSYCOS scales well on the large datasets, and obtain a linear speedup w.r.t. the available resources.

%% file: conclusion.tex
\section{Conclusion and Future Work}\label{sec:conclusion}
In the present paper, we propose an integrated SYnchronous COrrelation Search (iSYCOS) framework to find multi-scale temporal correlations in big time series. iSYCOS integrates both the top-down approach SYCOS$_{\text{TD}}$ and the bottom-up approach SYCOS$_{\text{BU}}$ into a single framework, and is thus able to efficiently extract various types of correlation relations, including both linear and non-linear, monotonic and non-monotonic, functional and non-functional, from time series. Our main contributions are the two search approaches that fit to specific scenarios (top-down for strongly correlated data, and bottom-up for weakly correlated data), and the proposed novel MI-based theory to identify noise in the data, and the MI optimized computation technique to optimize the search process. We also propose a robust selection algorithm that can automatically analyze and select the most efficient search approach for a given input time series. We further design a distributed version of iSYCOS that scales well on a Spark cluster to handle big time series. We perform an extensive experimental evaluation to evaluate the effectiveness and the performance of iSYCOS using both synthetic and real-world datasets. The evaluation shows that our method can capture various types of relations in synthetic data, and interesting correlations in real-world data. Moreover, the proposed optimizations are shown to be very efficient, and result in a speedup of up to an order of magnitude. The selection algorithm is robust on different datasets, and the distributed iSYCOS scales linearly to the data size and the number of computing nodes.

In future work, iSYCOS will be extended to capture correlations across spatial dimensions. The result of this work can also provide the foundation for advanced data analytics, such as to perform mining or infer causal effects from the extracted correlations.